\documentclass[11pt]{article}
\usepackage[english]{babel} % idioma del documento
\usepackage[utf8]{inputenc} % encoding del .tex
\usepackage[T1]{fontenc} % encoding del .pdf
\usepackage{lmodern} % para T1, la fuente default (Computer Modern) no funciona, hay que usar lmodern (por ejemplo)
\usepackage{amsmath}
\usepackage{amssymb}
\usepackage{amsthm}
\usepackage{color}
\allowdisplaybreaks[1] % (requiere amsmath) permitir saltos de pagina en ecuaciones (1: evitarlos si se puede, 4: se permiten siempre)
\usepackage{marginnote}
\usepackage[a4paper, left=2cm, right=2cm, top=3cm, bottom=3cm]{geometry}
\usepackage{titlesec}
\usepackage{hyperref}
\usepackage{graphicx}

\titleformat{\section}[block]{\Large\bfseries}{\thesection.}{3pt}{} % normal section
\titleformat{\subsection}[runin]{\normalfont\bfseries}{\thesubsection.}{3pt}{}[.] % subsection without vertical space after title

\def\RR{\mathbb{R}} % reals
\def\NN{\mathbb{N}} % naturals
\def\EE{\mathbb{E}} % expectation
\def\FF{\mathcal{F}} % an operator
\def\MM{\mathcal{M}} %space of bounded positive Borel measures.
\def\P{\mathcal{P}} % Poisson point measure
\def\Q{\mathcal{Q}} % another Poisson point measure
\def\R{\mathcal{R}} % yet another Poisson point measure
\def\ind{\mathbf{1}} % indicator
\def\ii{\mathbf{i}} % bold i
\def\Law{\textnormal{Law}} % law
\def\-{\text{-}} % small minus sign

\def\tp{\tilde{P}}%rescaled thermostat
\def\tpast{\tilde{P}^\ast}%adjoint of rescaled thermostat
\newtheorem{thm}{Theorem}
\newtheorem{lem}[thm]{Lemma}

\newtheorem{defi}[thm]{Definition}
\theoremstyle{definition}
\newtheorem{rmk}[thm]{Remark}

\title{On a thermostated Kac model with rescaling}

\author{Roberto Cortez\footnote{Universidad Andr\'es Bello, Departamento de Matem\'aticas. E-mail: \texttt{roberto.cortez.m@unab.cl}. Supported by Iniciaci\'on Fondecyt Grant 11181082 and by Programa Iniciativa Cient\'ifica Milenio through Nucleus Millenium Stochastic Models of Complex and Disordered Systems.}
\, and Hagop Tossounian\footnote{Universidad de Chile, DIM-CMM, E-mail: \texttt{htossounian@cmm.uchile.cl}. Supported by Fondecyt Postdoctoral
Project 3200130 and by Programa Iniciativa Cient\'ifica Milenio through Nucleus Millenium Stochastic Models of Complex and Disordered Systems.
}
}

\begin{document}

\maketitle

%%%%%%%%%%%%%%%%
% Abstract
%%%%%%%%%%%%%%%%

\begin{abstract}

We introduce a global thermostat on Kac's 1D model for the velocities of particles in a space-homogeneous gas subjected to binary collisions, also interacting with a (local) Maxwellian thermostat. The global thermostat rescales the velocities of all the particles, thus restoring the total energy of the system, which leads to an additional drift term in the corresponding nonlinear kinetic equation. We prove ergodicity for this equation, and show that its equilibrium distribution has a density that, depending on the parameters of the model, can exhibit heavy tails, and whose behaviour at the origin can range from being analytic, to being $C^k$, and even to blowing-up. Finally, we prove propagation of chaos for the associated $N$-particle system, with a uniform-in-time rate of order $N^{-\eta}$ in the squared $2$-Wasserstein metric, for an explicit $\eta \in (0, 1/3]$.

\end{abstract}



%%%%%%%%%%%%%%%%
% Introduction
%%%%%%%%%%%%%%%%

\section{Introduction}

\subsection{Thermostated Kac model}
In 1956, Marc Kac \cite{kac1956} introduced a space-homogeneous dynamics for a collection of $N$ identical particles with one-dimensional velocities $\mathbf{v} = (v_1, \dots, v_N) \in \RR^N$, undergoing binary collisions. It can be described as follows: select two particles $i\neq j$ at random, and update their velocities according to the rule
\begin{equation}
\label{eq:Kac_collision_rule}
(v_i,v_j)
\mapsto (v_i', v_j')
= (v_i \cos\theta - v_j\sin\theta, v_i \sin\theta+ v_j\cos\theta),
\end{equation}
for $\theta$ chosen randomly and uniformly on $[0, 2\pi)$. This rule preserves the total energy, i.e., $(v_i')^2+(v_j')^2 = v_i^2+v_j^2$. Particles evolve continuously with time $t\geq 0$ by iterating \eqref{eq:Kac_collision_rule}, where the times between collisions are independent and exponentially distributed with parameter $\lambda N$; thus, every particle undergoes $2\lambda$ collisions per unit of time on average. This is known as \emph{Kac's model}, and the above procedure produces a pure-jump Markov process on $\RR^N$, called the \emph{particle system}; we denote $f_t^N$ the probability distribution of the particle system at time $t\geq 0$. Note that the total energy of the system $\sum_i v_i^2$ is preserved; thus, if the vector of initial velocities belongs to $S^{N-1}(\sqrt{N E}) := \{\mathbf{v} \in \RR^N : \sum_i v_i^2 = NE \}$, for some fixed energy $E>0$, then $f_t^N$ is supported on $S^{N-1}(\sqrt{N E})$ for all $t\geq 0$.

Assume that the family of initial distributions $(f_0^N)_{N\in\NN}$ is symmetric and \emph{chaotic} with respect to some probability measure $f_0$ on $\RR$, that is, for any fixed $k$, the marginal distribution of $f_0^N$ in the first $k$ variables converges weakly, as $N\to\infty$, to $f_0^{\otimes k}$. In other words, initially, any fixed number of particles becomes asymptotically independent and $f_0$-distributed, as the total number of particles grows. Assuming that $f_0$ has a density, Kac proved that this property is propagated to later times by the dynamics: for all $t\geq 0$, the family $(f_t^N)_{N\in\NN}$ is symmetric and chaotic with respect to the density $f_t$, which is the solution to the so-called Boltzmann-Kac equation:
\begin{equation}
\label{eq:BoltzmannKac}
\partial_t f_t(v)
= 2\lambda \int_\RR \int_0^{2\pi} [f_t(v')f_t(v_*') - f_t(v)f_t(v_*)] \frac{d\theta}{2\pi} dv_*.
\end{equation}
This property is now termed \emph{propagation of chaos}, and it has been extensively studied during the last decades, for this and other related kinetic models, see for instance \cite{carlen-carvalho-leroux-loss-villani2010,cortez2016,mckean1966,mischler-mouhot2013}.

Call $\sigma_N$ the uniform distribution on $S^{N-1}(\sqrt{N E})$, and assume that $f_0^N$ is supported on $S^{N-1}(\sqrt{N E})$ and has a density with respect to $\sigma_N$. Then, for all $t\geq 0$, $f_t^N$ also has a density with respect to $\sigma_N$, and the particle system is ergodic in $L^1(S^{N-1}(\sqrt{N E}), \sigma_N)$, having $1$ as the unique equilibrium distribution. Kac initially worked in $L^2(S^{N-1}(\sqrt{N E}), \sigma_N)$ and conjectured that there is a spectral gap bounded below uniformly in $N$. This conjecture was proven by Janvresse in \cite{janvresse2001} and the gap was computed explicitly by Carlen, Carvalho, and Loss in \cite{carlen-carvalho-loss2003}. Since convergence in $L^2$ gives a crude upper bound to convergence in $L^1$, the relative entropy $\int_{S^{N-1}(\sqrt{NE})}f_t^N \log f_t^N \sigma_N(d\mathbf{v})$ was studied. The best known convergence results for the relative entropy are exponential with rate of order $1/N$. To ensure fast convergence to equilibrium in $L^1$, the authors in \cite{bonetto-loss-vaidyanathan2014} looked for systems close to equilibrium. More specifically: for a larger system of $\mathcal{N}$ particles (where $\mathcal{N} \gg N \gg 1$), they considered initial states of the form $f^\mathcal{N}_0(\mathbf{v}) = l_N(v_1, \dots, v_N) \prod_{j=N+1}^\mathcal{N}\gamma(v_j)$,
where $l_N$ is a density on $L^1(\RR^N)$, and
\begin{equation}
    \label{eq:gaussian}
    \gamma(w) = (2\pi T)^{-1/2} e^{-w^2/(2T)}
\end{equation}
is the Gaussian at a fixed temperature $T$. They argued heuristically that for the first $N$ particles in the system, this effectively works as an external thermostat. Thus, they introduced the \emph{thermostated Kac model} for a system of $N$ particles, in which, in addition to the collisions mentioned above, every particle interacts with a thermostat at a rate $\mu$. That is: when the thermostat acts on particle $i \in \{1,\ldots,N\}$, its velocity $v_i$ gets replaced by $v_i\cos\theta- w \sin\theta$, and the rest of the velocities are not affected. Here $w$ is chosen independently and randomly with density $\gamma(w)$, and $\theta$ is again chosen uniformly on $[0,2\pi)$. We will refer to this rule as the \emph{weak thermostat}. As in Kac's original model, propagation of chaos also holds in this case, as first shown in \cite{bonetto-loss-vaidyanathan2014}; see also \cite{cortez-tossounian2019} for a quantitative result.

In this paper we consider the \emph{strong thermostat}, also known as the \emph{Andersen thermostat} and as the \emph{Maxwellian thermostat}, which, when acting on particle $i$, simply replaces velocity $v_i$ by $w$, sampled with density $\gamma$. The corresponding nonlinear limit equation in this case is
\begin{equation}
\label{eq:thermostatedBoltzmannKac}
\partial_t f_t(v)
= 2\lambda \int_\RR \int_0^{2\pi} [f_t(v')f_t(v_*') - f_t(v)f_t(v_*)] \frac{d\theta}{2\pi} dv_*
+\mu[\gamma(v) - f_t(v)].
\end{equation}
We note that the strong thermostat is related to the weak one through a van Hove limit, see \cite{tossounian-vaidyanathan2015} for details.

Interactions with the (weak or strong) thermostat no longer preserve the total energy, thus the velocity distribution for the system is supported on the whole space $\RR^N$. This model is ergodic in $L^1(\RR^N)$ and has the Gaussian $\prod_{i=1}^N \gamma(v_i)$ as its unique equilibrium distribution. The density $f_t^N$ converges exponentially fast to equilibrium in relative entropy $\int_{\RR^N} f_t^N \log( f_t^N / \prod_i \gamma(v_i) ) d\mathbf{v}$ with an exponent bounded below uniformly in $N$; see \cite{bonetto-loss-vaidyanathan2014}. We mention that the argument behind the model with thermostats was made rigorous in \cite{bonetto-loss-tossounian-vaidyanathan2017} as a mean-field limit of finite-reservoir systems using the $L^2$ metric and the Fourier-based Gabetta-Toscani-Wennberg metric $d_2$. Studying the convergence to equilibrium for the finite-reservoir system is not a simple problem (see  \cite{bonetto-geisinger-loss-ried2018}).

\subsection{Global thermostats and rescaling}
\label{sec:thermostat_with_rescaling}
 The thermostats mentioned above act \emph{locally}, in the sense that their action on one particle is not connected to their action on the rest of the particles. In this work, we plan on introducing \emph{global} thermostats to Kac's model which affect the velocities of all the particles of the system in a connected way. An example of a global thermostat is the Gaussian isokinetic thermostat (see \cite{carlen-mustafa-wennberg2015}, \cite{morriss-dettmann1998}). In \cite{bonetto-carlen-esposito-lebowitz-marra2014},
this global thermostat was applied to a space-dependent particle model with an external electric field, where it prevented the energy of the system from growing to infinity and introduced non-equilibrium steady states.

Another kind of global thermostat is obtained by \emph{rescaling} the velocity of each particle after an interaction against the external source, thus keeping the energy of the system constant. Global thermostats with rescaling have been used in molecular dynamics. An example is the Berendsen thermostat \cite{berendsen-postma-gunsteren-dinola-haak1984} which was used to study particle systems that have temperature and pressure as controllable parameters, instead of the total energy and volume  (see \cite{bussi-parrinello2008} and references therein).

In this work, we will consider the Kac particle system coupled to the strong thermostat where, in addition, we allow the system to restore its total energy towards $NE$ by introducing a rescaling mechanism on this model (hence making our thermostats global). One of our motivations for introducing this model is to obtain a system that retains some of the nice equilibration features of the (locally) thermostated particle system, while keeping the energy of the system closer to its initial value, as in the original Kac model.

More specifically, one rescaling mechanism that we will consider is the following: after a strong thermostat interaction of particle $i$ against a sample $w$ taken from the Gaussian distribution \eqref{eq:gaussian}, the velocity of \emph{every particle} in the system is multiplied by the factor
\begin{equation}
\label{eq:alphaN}
    \alpha_N(v_i, w) = \sqrt{\frac{NE}{NE- v_i^2+w^2}}.
\end{equation}
This has the effect of keeping the total energy $\sum_i v_i^2$ constant and equal to $NE$ almost surely. The corresponding finite particle system is defined rigorously through its generator in \eqref{eq:masterAlpha} in the Appendix. We remark that we will mostly be working with a particle system defined using a slightly different rescaling factor, see \eqref{eq:betaN} below. In the sequel, we will use the term \emph{thermostat} to refer to the local thermostat dynamics described by \eqref{eq:thermostatedBoltzmannKac}, while the term \emph{global thermostat} will refer to the rescaling mechanism.

Note that, for $N$ large, the factor $\alpha_N(v_i, w)$ becomes approximately 1; 
making the effect of rescaling after a single thermostat interaction negligible. However, since there are order $N$ such interactions on every finite time interval, the effect of the rescaling mechanism survives as $N\to\infty$, giving rise to an additional \emph{drift} term in \eqref{eq:thermostatedBoltzmannKac}, namely  $-\mu \frac{E-T}{2E} \partial_v (vf_t(v))$. This argument will be made rigorous in Lemma \ref{lem:chaosLN}. The Kac and thermostat interactions produce the same terms as in \eqref{eq:thermostatedBoltzmannKac}. Consequently, the corresponding candidate limit equation for our thermostated Kac model with rescaling is
\begin{equation}
\label{eq:PDE}
\begin{split}
    \partial_t f_t(v)
    &= 2\lambda \int_\RR \int_0^{2\pi} [f_t(v')f_t(v_*') - f_t(v)f_t(v_*)] \frac{d\theta}{2\pi} dv_* \\
    & \quad {} + \mu [\gamma(v) - f_t(v)]
    -A \partial_v (vf_t(v)),
\end{split}
\end{equation}
for the collection $(f_t)_{t\geq 0}$ of probability densities on $\RR$, where we introduced the constant
\[
A = \mu \frac{E-T}{2E}.
\]

The present paper can be seen as a continuation of our previous work \cite{cortez-tossounian2019}, where we used similar tools to make quantitative the propagation of chaos result for the thermostated Kac model (without rescaling) in \cite{bonetto-loss-vaidyanathan2014}.
Here, the main objects of study are the thermostated finite particle system with rescaling and its corresponding kinetic equation \eqref{eq:PDE}. To the best of our knowledge, the rescaling mechanism and its associated drift term are new in the context of Kac-type kinetic models. The drift term in \eqref{eq:PDE} has the effect of taking the energy of $f_t$ towards $E$, see Theorem \ref{thm:WellPosedness2}.
We remark that, apart from the physical motivation explained above, the nonlinear equation \eqref{eq:PDE} is mathematically interesting on its own. Indeed, as we shall see, the equation is contractive and admits an equilibrium density having properties that differ from those of the Gaussian: it can exhibit heavy tails or explosion at the origin, depending on the parameters $\lambda$, $\mu$, $E$, and $T$.

\subsection{Main results}

Our main results concern the convergence properties of the nonlinear equation \eqref{eq:PDE}, and the propagation of chaos for its corresponding particle system. We will use the following distance to quantify convergence: given $\nu$ and $\tilde{\nu}$ probability measures on $\RR^k$, their 2-\emph{Wasserstein distance} is given by
\begin{equation}
\label{eq:W2}
    W_2(\nu,\tilde{\nu})
    = \left( \inf \EE\left[ \frac{1}{k} \sum_{i=1}^k |X_i - Y_i|^2 \right]\right)^{1/2},
\end{equation}
where the infimum is taken over all \emph{couplings} of $\nu$ and $\tilde{\nu}$, that is, over all random vectors $\mathbf{X} = (X_1,\ldots,X_k)$ and $\mathbf{Y} = (Y_1,\ldots,Y_k)$ defined on some common probability space, such that $\Law(\mathbf{X}) = \nu$ and $\Law(\mathbf{Y}) = \tilde{\nu}$. We follow the convention of using the  normalizing factor $\frac{1}{k}$ to simplify the dependence on dimension. The infimum in \eqref{eq:W2} is always achieved by some $(\mathbf{X},\mathbf{Y})$, and such a pair is called an \emph{optimal coupling}, see \cite{villani2002} for details.

Regarding the nonlinear equation \eqref{eq:PDE}, we are interested in its stability properties, which we study in Section \ref{sec:PDE}. We define a notion of weak solution for continuous functions of $t\geq 0$ taking values in the space of probability measures, and we prove the existence and uniqueness of a solution $(f_t)_{t\geq 0}$ with bounded energy, see Definition \ref{def:weak_solutions} and Theorem \ref{thm:WellPosedness2}. In Lemma \ref{lem:moments} we prove the stability of the $r^\text{th}$-moments of $f_t$, for all $r>2$ satisfying
\begin{equation}
	\label{eq:unif-r-moment}
	r A < 2\lambda C_r + \mu,
	\end{equation}
where $C_r := 1 - 2\int_0^{2\pi} |\!\cos\theta|^r \frac{d\theta}{2\pi}>0$ (note that there always exists such an $r>2$, as long as $2\lambda+\mu>0$). The existence of an equilibrium distribution and its main properties are stated in the following result, proven at the end of Section \ref{sec:PDE}:

\begin{thm}[equilibrium distribution]
\label{thm:continuity}
Let $2\lambda+\mu>0$. There exists a probability measure $f_\infty$ on $\RR$ with energy $E$, such that $\lim_{t\to \infty} f_t = f_\infty$ weakly, for any $(f_t)_{t\geq 0}$ weak solution of \eqref{eq:PDE} having bounded initial energy. Moreover,
\[
    W_2(f_t,f_\infty)
    \leq W_2(f_0,f_\infty) e^{- \frac{\mu T}{2E} t},
\]
and $f_\infty$ is the unique stationary weak solution of \eqref{eq:PDE}. Moreover, $f_\infty(dv)$ has a density,  denoted by $f_\infty(v)$, that has the following properties:
\begin{enumerate}
    \item $f_\infty$ is even and monotone decreasing on $(0,\infty)$.
    
    \item $\int_\RR f_\infty(v) \vert v\vert^r dv < \infty$ if and only if $r$ satisfies condition \eqref{eq:unif-r-moment}.
    
    \item $f_\infty \in C^\infty(\RR\backslash \{0\})$, and satisfies for all $v \in \RR\backslash\{0\}$:
    \begin{equation}
        \label{eq:PDEstationary-strong}
        A \partial_v (v f_\infty(v))
        = 2\lambda (B_2[f_\infty, f_\infty](v) - f_\infty(v)) + \mu [\gamma(v) - f_\infty(v)].
    \end{equation}
    
    \item
    \label{item:p-derivatives}
    Regarding the behaviour of $f_\infty$ at the origin, we have: for any integer $p\geq 0$, $f_\infty \in C^p(\RR)$ if and only if
    \begin{equation}
    \label{eq:Cp-Condition}
    \frac{T}{E} < 1 + \frac{2}{p+1}\left( \frac{\mu}{2\lambda+\mu} \right)^{-1}.
    \end{equation}
    Consequently, if $\frac{T}{E} \geq 1 + 2\left( \frac{\mu}{2\lambda+\mu} \right)^{-1}$, then $\lim_{v\to 0}f_\infty(v) = \infty$. Moreover, if $E\geq T$, then $f_\infty$ is analytic. This also holds for non-integer $p>0$, if we interpret the statement  $f_\infty \in C^p(\RR)$ as $\int_\RR \vert \hat{f}_\infty (\xi)\vert  \vert \xi\vert^p \,d\xi <\infty$, where $\hat{f}_\infty$ denotes the Fourier transform of $f_\infty$.
\end{enumerate}
\end{thm}

Note that parts (ii) and (iv) of this theorem imply that unless we have $E=T$ or $\mu = 0$ (in which case $f_\infty=\gamma$), there are two cases: either $f_\infty$ does not have finite moments of high order (case $E>T$), or $f_\infty$ cannot have too many derivatives at the origin (case $E<T$). We summarize the smoothness properties of $f_\infty$ at the origin in Figure \ref{fig:phase-diagram}.

\begin{figure}[th]
    \centering
    \includegraphics[width=0.7\textwidth]{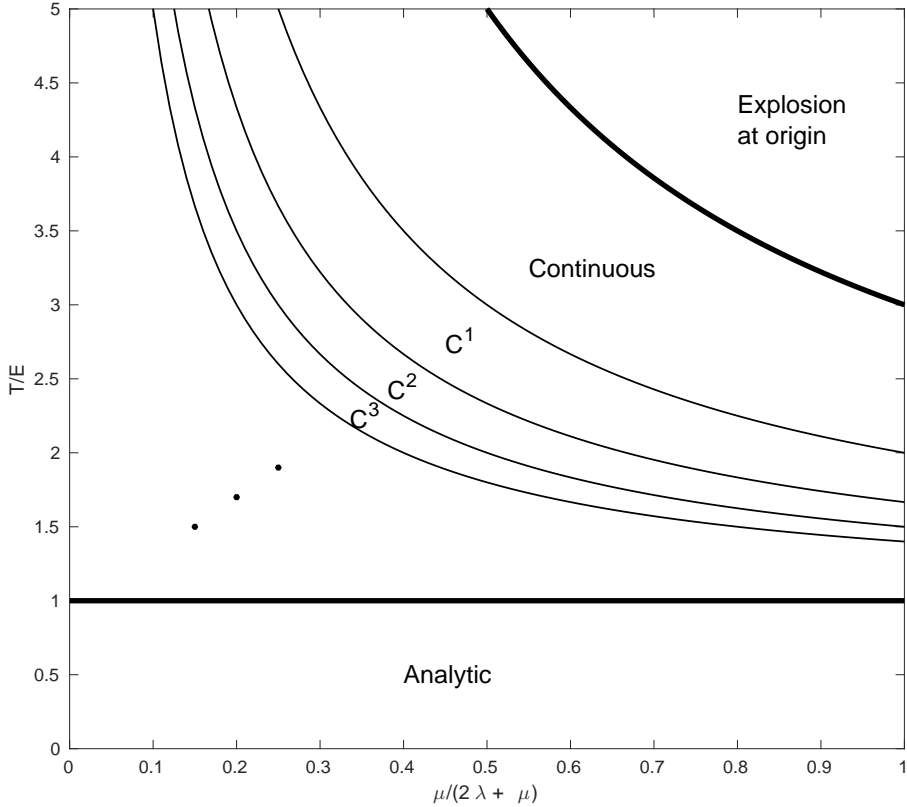}
    \caption{Phase diagram for the smoothness of $f_\infty$ at the origin.}
    \label{fig:phase-diagram}
\end{figure}

In Section \ref{sec:particle-system}, we turn to the study of the particle system associated with \eqref{eq:PDE} and its propagation of chaos property. Unfortunately, the rescaling factor $\alpha_N(v_i,w)$ in \eqref{eq:alphaN} is hard to handle mathematicaly, since it takes large values on sets where some $v_i^2$ is close to $NE$. Consequently, for the particle system with $\alpha_N$ as the rescaling factor, we could only prove propagation of chaos at the level of the generators, see Lemma \ref{lem:chaosLN}. In order to obtain a more tractable particle system, we will consider a different rescaling factor, namely 
\begin{equation}\label{eq:betaN}
\beta_N(w) = \sqrt{\frac{NE}{NE-E+w^2}},
\end{equation}
which is obtained from $\alpha_N(v_i,w)$ by replacing $v_i^2$ by its expected value $E$. We call\footnote{We do not make explicit the dependence of the particle system on the number of particles $N$. Also, we use superscripts for the particles' indexes.} $\mathbf{V}_t = (V_t^1,\ldots,V_t^N)$ the Markov process with Kac and thermostat interactions, rescaled using $\beta_N$, as described in Section \ref{sec:thermostat_with_rescaling}; see \eqref{eq:sde_ps} in Section \ref{sec:particle-system} for an explicit construction using an SDE with respect to a Poisson point measure. Even though this particle system no longer preserves the total energy exactly, the expected total energy is preserved, see Lemma \ref{lem:EV-squared}. We also prove a contraction property in Lemma \ref{lem:contraction_ps}, and deduce equilibration in Lemma \ref{lem:equilibration_ps}.

Before we state our second main result, let us recall the following characterization of chaoticity, see \cite[Proposition 2.2]{sznitman1989} for details: for fixed $t\geq 0$ the collection $(f_t^N)_{N\in\NN}$, with $f_t^N = \Law(\mathbf{V}_t)$, is $f_t$-chaotic if and only if the empirical (random) measure of the system
\[
\bar{\mathbf{V}}_t := \frac{1}{N} \sum_{i=1}^N \delta_{V_t^i}
\]
converges weakly to the (non-random) measure $f_t$.

\begin{thm}[propagation of chaos]
\label{thm:UPoC}
Fix $r>2$ such that condition \eqref{eq:unif-r-moment} is satisfied, and assume that $\int |v|^r f_0(dv) < \infty$. Then  there is a finite constant $C$ depending only on $\lambda$, $\mu$, $E$, $T$, $r$, and $\int |v|^r f_0(dv)$, such that for all $t\geq 0$ we have:
\begin{equation}
\label{eq:UPoC}
    \EE[W_2^2(\bar{\mathbf{V}}_t, f_t)]
    \leq 4 e^{-\frac{\mu T}{E} t} W_2^2(f_0^N, f_0^{\otimes N})
     + C
     \begin{cases}
        N^{-1/3}, & r>4, \\
        N^{-1/3}(\log N)^{2/3}, & r=4, \\
        N^{- \frac{r-2}{2(r-1)}}, &  r \in (2,4).
    \end{cases}
\end{equation}
\end{thm}

Thus, this result, proved in Section \ref{sec:particle-system}, establishes the propagation of chaos for the thermostated Kac particle system with rescaling, with an explicit uniform-in-time rate of order $N^{-1/3}$ in expected squared 2-Wasserstein distance, when $\int |v|^r f_0(dv) < \infty$ and condition \eqref{eq:unif-r-moment} is satisfied for some $r>4$. This is not so far from the optimal general rate $N^{-1/2}$, valid for the expected squared 2-Wasserstein distance between the empirical measure of an i.i.d.\ sequence and its common law, see \cite[Theorem 1]{fournier-guillin2013}.

To obtain an explicit chaos rate, one needs to ensure that the initial condition term $W_2(f_0^N, f_0^{\otimes N})^2$ converges to 0 as $N\to\infty$, hopefully with rate $N^{-1/3}$ or faster. For instance, one can simply take $f_0^N = f_0^{\otimes N}$, thus $W_2(f_0^N, f_0^{\otimes N})=0$. However, we would like Theorem \ref{thm:UPoC} to be useful when the initial distributions $f_0^N$ are supported on the sphere $S^{N-1}(\sqrt{NE})$, as it is commonly considered in the literature of Kac-type models. One method of obtaining such an $f_0$-chaotic sequence, used for instance in \cite{carlen-carvalho-leroux-loss-villani2010,carrapatoso2015}, is to \emph{condition} $f_0^{\otimes N}$ to the sphere; see for instance \cite[Theorem 1.5]{hauray-mischler2014} and \cite[Equation (5.33)]{carrapatoso-einav2013} for a specific rate.

In our next result, we provide an alternative approach, based on \emph{scaling}. Specifically: given a probability measure $f$ on $\RR$ satisfying $\int v^2 f(dv)= E$, take a random vector $\mathbf{X} = (X_1,\ldots,X_N)$ with distribution $f^{\otimes N}$, and define
\begin{equation}
\label{eq:fN}
f^N = \Law(\mathbf{Y}),
\quad \text{with $\mathbf{Y} = (Y_1,\ldots,Y_N)$ given by} \quad
Y_i = \left( \frac{E}{ \frac{1}{N}\sum_j X_j^2} \right)^{1/2} X_i
\end{equation}
on the event $\sum_i X_i^2>0$, while on the event $\sum_i X_i^2=0$ the vector $\mathbf{Y}$ is defined as some given exchangeable random vector taking values on $S^{N-1}(\sqrt{NE})$, independent of $\mathbf{X}$ (its particular distribution is irrelevant). By construction, $f^N$ is supported on the sphere $S^{N-1}(\sqrt{NE})$.

The following theorem provides a quantitative chaoticity estimate valid for general measures; its proof is given at the end of Section \ref{sec:particle-system}. To state our result, we introduce the following: given $N\in\NN$ and a probability measure $f$ on $\RR$, define
\begin{equation}
\label{eq:epsN}
\varepsilon_N(f)
= \EE W_2(\bar{\mathbf{X}}, f)^2,
\end{equation}
where $\mathbf{X} = (X_1,\ldots,X_N)$ has distribution $f^{\otimes N}$, and $\bar{\mathbf{X}} = \frac{1}{k} \sum_{i=1}^k \delta_{X^i}$ is its (random) empirical measure. The best available rates of convergence for $\varepsilon_N(f)$ are found in \cite[Theorem 1]{fournier-guillin2013}, which we present in our setting in \eqref{eq:eps}; we use these estimates in the Theorem below. We remark that the scaling procedure \eqref{eq:fN} and the main idea of the proof are already present in \cite[Lemma 25]{fournier-guillin2015} and \cite[eq.\ (20)]{cortez2016}, where explicit rates of chaoticity for the uniform measures on spheres towards the Gaussian are obtained.

\begin{thm}
\label{thm:chaoticSN}
Let $f$ be a probability measure on $\RR$ with $\int v^2 f(dv)= E$ and having finite $r^{\text{th}}$-moment for some $r>2$. Then, $(f^N)_{N\in\NN}$ given by \eqref{eq:fN} is a sequence of distributions supported on the sphere $S^{N-1}(\sqrt{NE})$ which is chaotic to $f$. Moreover, we have
\[
W_2^2(f^N, f^{\otimes N})
\leq \varepsilon_N(f).
\]
Consequently (see \eqref{eq:eps}), we obtain the following quantitative rate: there exists a constant $C$ depending only on $r$ and $\int |v|^r f(dv)$, such that
\[
W_2(f^N, f^{\otimes N})^2 \leq C
\begin{cases}
    N^{-1/2}, & r>4, \\
    N^{-1/2}\log(N), & r=4, \\
    N^{- (1-\frac{2}{r})}, &  r \in (2,4).
\end{cases}
\]
\end{thm}

We stress the fact that the scaling \eqref{eq:fN} and the statement of Theorem \ref{thm:chaoticSN} make no assumptions about $f$, except that it has a finite $r^\text{th}$-moment for some $r>2$. In particular, $f$ is not assumed to have a density, which is required to even define the conditioned tensor product. We also note that the chaoticity estimate we provide is a ``strong'' one, in the sense that it compares the full distributions $f^N$ and $f^{\otimes N}$, instead a fixed number of marginals, or the empirical measure of a sample of $f^N$ against $f$, as in Theorem \ref{thm:UPoC}.

\subsection{Plan of the paper}

Our proofs use both analytic and probabilistic tools. The proof of the main properties of the equilibrium distribution $f_\infty$, stated in Theorem \ref{thm:continuity}, relies on the Fourier transform. To obtain our contraction and chaoticity estimates in Wasserstein distance, we will make use of probabilistic coupling arguments. Most notably, in Section \ref{sec:PDE} we will introduce the \emph{Boltzmann process},\footnote{The term ``Boltzmann process'' is typically used in the context of the Boltzmann equation. We will use the same name for our thermostated Kac model with rescaling.} a stochastic process having marginal laws equal to $f_t$, and use it throughout the rest of the paper to prove many of our convergence results.

We remark that our results also hold true in the case of the more physical weak thermostat mentioned above. Most of the the proofs can be easily adapted to this case. However, some intermediate lemmas, such as Lemma \ref{lem:app}, become quite technical. To keep the exposition simple, we chose to work with the strong thermostat.

The structure of the paper is as follows: in Section \ref{sec:PDE} we study the nonlinear equation \eqref{eq:PDE} and its properties. In Section \ref{sec:particle-system} we study the particle system and prove the propagation of chaos.
Finally, in the Conclusion we mention some open problems. The proofs of some technical lemmas are left to the Appendix.

%%%%%%%%%%%%%%%%
% PDE
%%%%%%%%%%%%%%%%
\section{Analysis of the nonlinear equation}
\label{sec:PDE}

\subsection{Well-posedness} 

In this section we study the well-posedness of \eqref{eq:PDE} and its main properties. We now introduce some useful notation. Let  $\mathcal{B}(\RR)$ denote the space of signed, finite Borel measures, and define the mapping $B_2: \mathcal{B}(\RR) \times \mathcal{B}(\RR) \rightarrow \mathcal{B}(\RR)$ via the equation
\begin{equation}\label{eq:B2}
  \int_\RR B_2[\nu, \tilde{\nu}](dv)\phi(v) = \int_\RR \int_\RR \int_0^{2\pi} \phi(v\cos\theta+ v_* \sin\theta) \frac{d\theta}{2\pi} \nu(dv)\tilde{\nu}(dv_*),
\end{equation}
for all  continuous and bounded functions $\phi$ and all $\nu, \tilde{\nu}$ in $\mathcal{B}(\RR)$. We note that $B_2[\cdot,\cdot]$ is a symmetric bilinear form on $\mathcal{B}(\RR)$ satisfying $\Vert B_2[\nu, \tilde{\nu}] \Vert \leq \Vert \nu \Vert \Vert \tilde{\nu} \Vert$, where $\Vert \cdot \Vert$ denotes the total variation norm. $B_2$ has the following monotonicity property for non-negative measures:
\[
    \nu \leq \tilde{\nu} 
    \quad \Rightarrow \quad
    B_2[\nu, \nu] \leq B_2[\tilde{\nu},\tilde{\nu}],
\]
where the inequality $\nu \leq \tilde{\nu}$ means $\nu(A) \leq \tilde{\nu}(A)$ for all measurable set $A$. If $\nu, \tilde{\nu}$ have densities $f(v),g(v) \in L^1(\RR)$, then $B_2[\nu,\tilde{\nu}]$ also has a density which we also denote by $B_2[f,g]$. In this case, $B_2[f,g]$ satisfies
 \[
    B_2[f, g](v) = \int_\RR \int_0^{2\pi} f(v')g(v_*') \frac{d\theta}{2\pi}dv_*.
 \]
Recall that $A = \mu \frac{E-T}{2E}$, $\gamma(w) = (2\pi T)^{-1/2} e^{-w^2/(2T)}$ is the thermostat distribution at temperature $T$, and $E>0$ is the average $1$-particle energy in our system. Thus, our nonlinear equation \eqref{eq:PDE} for the collection of densities $(f_t)_{t\geq 0}$ takes the form
\[
    \partial_t f_t(v)
    = 2\lambda (B_2[f_t, f_t]- f_t) + \mu [\gamma(v) - f_t(v)]  - A \partial_v (vf_t(v)),
\]
where $2\lambda$ and $\mu$ are the rates of Kac and thermostat interactions, respectively. We assume that $\int v^2 f_0(v)dv < \infty$. If $f_t$ has a density, we will also call it $f_t$, and we will use $f_t(v)dv$ and $f_t(dv)$ interchangeably (same for $\gamma(dv)$, etc.).
 
We now turn to the problem of existence and uniqueness of solutions for \eqref{eq:PDE}. Let us first write the equation in weak form. The key idea is to note that if $f_t(v)$ solves \eqref{eq:PDE}, then a formal computation shows that $e^{At} f_t(e^{At}v)$ solves an equation with no drift term. In order to provide a precise definition, for $t\geq 0$, let $B_1[t]$ be given by 
\begin{equation}
\label{eq:B1}
    B_1[t](dv)
    = e^{At} \gamma(e^{At}v) dv.
\end{equation}
Let $\P(\RR) \subseteq \mathcal{B}(\RR)$ denote the space of Borel probability measures on $\RR$, metrized with $\Vert \cdot \Vert$.

\begin{defi}[weak solution]
\label{def:weak_solutions}
We say that a collection $(f_t)_{t\geq 0} \in C([0,\infty),\P(\RR))$ is a weak solution to \eqref{eq:PDE}, if $(g_t)_{t\geq 0}$ defined by $\int \phi(v) g_t(dv) = \int \phi(e^{-At} v) f_t(dv)$ satisfies
\begin{equation}
\label{eq:PDEweak}
g_t
= f_0 + \int_0^t \left\{ 2\lambda (B_2[g_s,g_s] - g_s)
+ \mu (B_1[s] - g_s)
\right\} ds, \qquad \forall t\geq 0.
\end{equation}
\end{defi}

Note that $g$ satisfies \eqref{eq:PDEweak} if and only if 
\begin{equation}\label{eq:PDEweakNiceForm}
    g_t = e^{-Dt}f_0 + e^{-Dt} \int_0^t e^{Ds}\left\{ 2\lambda B_2[g_s,g_s] + \mu B_1[s]\right\} ds,
\end{equation}
where $D := 2\lambda + \mu$.
 We can now state and prove:

    \begin{thm}[well-posedness]
    \label{thm:WellPosedness2}
        For any $f_0 \in \P(\RR)$, there is a unique weak solution $(f_t)_{t\geq 0}$ to \eqref{eq:PDE}. If $f_0 \in L^1(\RR)$, then $f_t \in L^1(\RR)$ for all $t\geq 0$. If $f_0$ has finite $r^{\text{th}}$ moment, then so does $f_t$ for all $t$. In addition, if $\int_\RR v^2 f_0(dv)< \infty$, then the solution $f_t$ to \eqref{eq:PDEweak} satisfies
\[
    \int_\RR v^2 f_t(dv) = e^{-\frac{\mu T}{E}} \int_\RR v^2 f_0(dv) + E \left(1-e^{-\frac{\mu T}{E}}\right),
\]
for all future times. In this case, we see that
\begin{equation}\label{eq:energy.t}
    \int_\RR v^2 f_t(dv)\leq \max\left\{ E , \int_\RR v^2 f_0(dv)\right\}.
\end{equation}
    \end{thm}

\begin{proof}[Proof of Theorem \ref{thm:WellPosedness2}] 
The proof of existence and uniqueness for \eqref{eq:PDEweakNiceForm} uses the iterative construction of \cite{wild1951} (see also \cite{tanaka.s1968} and \cite{carlen-carvalho-gabetta2000}); which is standard. Here we provide a full proof for convenience of the reader. Let $\MM$ be the space of bounded, positive, Borel measures on $\RR$. Let $f_0$ be a Borel probability measure on $\RR$. Define the sub-probability measures $( u^n_t)_{n=0}^\infty$ inductively by
\begin{eqnarray}
    u^0_t & = & e^{-D t} f_0, \nonumber\\
    u^{n+1}_t & = & e^{-D t} f_0 + \int_0^t e^{-D (t-s)} \left( 2\lambda B_2[u^n_s, u^n_s] + \mu B_1[s] \right) ds \label{eq:iteration}
\end{eqnarray}
From the monotonicity properties of $B_2$, we see by induction that $u^n_t$ is continuous in $t$ for each $n$, that $u^n_t(\RR) \leq 1$, and that $(u^n_t)_n$ is increasing in $n$. Hence, for each $t$, we can define $u_t \in \MM$ as $u_t(A) = \lim_n u_t^n(A)$ for each measurable set $A$. Note that $u_t-u_t^n$ is a non-negative measure for each $t$, thus we have convergence in total variation, since
\[
    \lim_{n\rightarrow \infty} \Vert u^n_t - u_t \Vert =  \lim_{n\rightarrow \infty}  u_t(\RR) - u^n_t(\RR) = 0.
\]
This, together with the bilinearity of $B_2$ and its property that $\Vert B_2[ \nu_1, \nu_2] \Vert \leq \Vert \nu_1 \Vert \Vert \nu_2 \Vert$ for all signed measures $\nu_1, \nu_2$, implies that
\[
     \lim_{n\rightarrow \infty} \Vert B_2[u^n_t, u^n_t] - B_2[u_t, u_t] \Vert \leq \lim_{n\rightarrow \infty} \Vert u^n_t - u_t \Vert (u^n_t(\RR) + u_t(\RR)) = 0.
\]
Thus we can take the infinite $n$ limit in \eqref{eq:iteration} and establish that $u_t$ solves \eqref{eq:PDEweakNiceForm}. Being an increasing limit of continuous functions, $u: [0,\infty) \rightarrow \MM$ is lower semi-continuous, and thus measurable. Since $u_t(\RR) \leq 1$, $\forall t$, $u$ belongs to $L^\infty( [0,\infty), \MM)$. This fact, and \eqref{eq:PDEweakNiceForm} imply that $u_t$ is continuous in $t$. 

To prove that $u_t$ is a probability measure, we follow \cite{tanaka.s1968}. Using \eqref{eq:PDEweak}, we see that function $h(t)= u_t(\RR)$ is differentiable and satisfies the differential equation
\[
    h'(t) = -(2\lambda+\mu) h(t) + \mu + 2\lambda h(t)^2.
\]
Since $h(0)=1$, $h(t) \equiv 1$ for all $t$ as desired. To show the uniqueness of $u_t$, let $g_t \in C([0,\infty), \MM)$ satisfy \eqref{eq:PDEweakNiceForm}.
On one hand, $g_t\geq u^0_t$ by definition. And thus, by induction, $g_t \geq u^n_t$ a.e.\ $t$ for all $n$. Hence, we have $g_t \geq u_t$.
Since $g_t$ is continuous in $t$ by hypothesis, it must be a probability measure for all $t$ just like $u_t$. Thus, $\Vert g_t - u_t \Vert = g_t(\RR)- u_t(\RR) =0$.

We now show the propagation of being a density: let $f_0 \in L^1(\RR)$, then $u_t^n \in L^1(\RR)$ for all $\RR$ and we use the completeness of $L^1$ under the total variation norm. 
To prove the propagation of the $r^{\text{th}}$-moments, we first note that for any probability measures $\nu_1$ and $\nu_2$ with finite $r^{\text{th}}$ moments $n_r$ and $m_r$, $B_2[\nu_1, \nu_2]$ satisfies
\[
\int_\RR \vert x\vert^r B_2[\nu_1, \nu_2](dx) \leq  C_r \frac{n_r + m_r}{2}
\]
where $C_r = 2^{\max\{\frac{r}{2},1\}}\int_0^{2\pi} \vert \cos\theta\vert^r \frac{d\theta}{2\pi}$.
If $f_0$ has a finite $r^{\text{th}}$ moment for some $r>0$ then, by induction, we see that, for each $t$, $(\int_\RR u_t^n(dv) \vert v\vert^r)_n$ is finite, monotone increasing, and bounded above by the solution $R(t)$ to the following integral equation:
\[
    R(t) = e^{-(2\lambda + \mu) t} \int_\RR \vert v\vert^r f_0(dv) + C_r\int_0^t e^{-(2\lambda+\mu)(t-s)} (2\lambda+ \frac{\mu}{2})R(s) ds + C_r \int_\RR \vert w\vert^r \gamma(dw).
\]
Here $C_r$ is as above. $R(t)$ is finite due to Gronwall's inequality. The monotone convergence theorem implies that $R(t)$ controls the $r^{\text{th}}$ moment of $f_t$.
\end{proof}

\begin{rmk} \label{rmk:property_P}
The proof above can be used to show the following result which will be useful in proving Theorem \ref{thm:continuity}: if the initial density $f_0$ is even and monotone non-increasing on $[0,\infty)$, then $f_t$ also has this properties, for all $t$. To justify this, it is enough to show that being even and monotone on $[0, \infty)$ is preserved under the map $g \mapsto B_2[g,g]$. Indeed, this fact allows us to show inductively that $u_t^{n}$ is even and monotone on $[0, \infty)$ for all $n$, which then passes to the limit in $n$. Let us show the missing detail: fix $g \in L^1(\RR)$ nonnegative, even, and monotone on $[0,\infty)$. Changing the parameter $\theta$, one can write the post-collisional velocities as $v' = \sqrt{v^2+v_*^2} \cos\theta$ and $v_*' = \sqrt{v^2+v_*^2} \sin\theta$; thus, for $0\leq v \leq w$ and using the evenness of $g$, we have
\begin{align*}
    B_2[g,g](v)
    &= \int_\RR \int_0^{2\pi} g\left(\sqrt{v^2+v_*^2} ~|\! \cos\theta|\right) g\left(\sqrt{v^2+v_*^2} ~|\! \sin\theta|\right) \frac{d\theta}{2\pi} dv_* \\
    &\geq \int_\RR \int_0^{2\pi} g\left(\sqrt{w^2+v_*^2} ~|\! \cos\theta|\right) g\left(\sqrt{w^2+v_*^2} ~|\! \sin\theta|\right) \frac{d\theta}{2\pi} dv_* \\
    &= B_2[g,g](w),
\end{align*}
that is, $B_2[g,g](v)$ is non-increasing on $[0,\infty)$ when $g$ is.
\end{rmk}

\subsection{The Boltzmann process}

For a given $f_0 \in \P(\RR)$, let $(f_t)_{t\geq 0} \in C([0,\infty),\P(\RR))$ be the unique weak solution of \eqref{eq:PDE} given by Theorem \ref{thm:WellPosedness2}. We will now construct a stochastic process $(Z_t)_{t\geq 0}$, called the \emph{Boltzmann process}, such that $\Law(Z_t) = f_t$ for all $t\geq 0$. This process is the probabilistic counterpart of \eqref{eq:PDE}, and it represents the trajectory of a single particle immersed in the infinite population. It was first introduced by Tanaka \cite{tanaka1978} in the context of the Boltzmann equation for Maxwell molecules.

Consider a Poisson point measure $\P(dt,d\theta,dz)$ on $[0,\infty) \times [0,2\pi) \times \RR $ with intensity $2\lambda dt \frac{d\theta}{2\pi} f_t(dz)$, and an independent Poisson point measure $\Q(dt,dw)$ on $[0,\infty) \times \RR$ with intensity $\mu dt \gamma(dw)$. Consider also a random variable $Z_0$ with law $f_0$, independent of $\P$ and $\Q$. The process $Z_t$ is defined as the unique solution, starting from $Z_0$, to the stochastic differential equation
\begin{equation}
\label{eq:SDE}
\begin{split}
dZ_t
&= \int_0^{2\pi} \int_\RR [Z_{t^\-} \cos\theta - z\sin\theta - Z_{t^\-}] \P(dt,d\theta,dz) \\
& \quad {} + \int_\RR [w - Z_{t^\-}] \Q(dt,dw)
  + A Z_t dt.
\end{split}
\end{equation}

Strong existence and uniqueness of solutions for this SDE is straightforward, since the rates of $\P$ and $\Q$ are finite on bounded time intervals. To show that $\Law(Z_t) = f_t$, the argument is classical: one first shows that $\ell_t := \Law(Z_t)$ solves a linearized weak version of \eqref{eq:PDE}, namely
\begin{align*}
    \partial_t \int_\RR  \phi(v) \ell_t(dv)
    &=  2 \lambda \int_\RR \phi(v) (B_2[\ell_t, f_t] - \ell_t)(dv) \\
    & \quad {} + \mu \int_\RR \phi(v)( \gamma - \ell_t) (dv)
    + A \int_\RR v \phi'(v) \ell_t(dv)
\end{align*}
for all bounded and continuous $\phi$ for which $v\phi'(v)$ is bounded and continuous. This can be seen to have a unique solution in a weak sense similar to \eqref{eq:PDEweakNiceForm}, because: $\nu\mapsto B_2[\nu,f_t]$ is non-expanding in total variation for all $t$. Since $f_t$ is also a solution of this linearized version, we must have that $\ell_t = f_t$.

Using the Boltzmann process \eqref{eq:SDE} one can easily study the evolution of moments for the kinetic equation \eqref{eq:PDE}. As the next lemma shows, when $T \geq E$ (negative drift term in \eqref{eq:SDE}), every moment is propagated uniformly in time; otherwise, only some moments propagate:

\begin{lem}[propagation of moments]
    \label{lem:moments}
	Let $r>2$, assume that $\int |v|^r f_0(dv) < \infty$, and let $(f_t)_{t\geq 0}$ be the weak solution to \eqref{eq:PDE}. Then $\sup_{t\geq 0} \int |v|^r f_t(dv) < \infty$ if and only if $r$ satisfies condition \eqref{eq:unif-r-moment}.
\end{lem}

\begin{proof}
	Call $h(t) = \EE |Z_t|^r = \int |v|^r f_t(dv)$, since $\Law(Z_t) = f_t$. Using It\^o calculus for jump processes and the inequality $(a+b)^r \leq a^r + b^r + 2^{r-1}(ab^{r-1} + ba^{r-1})$ for $a,b\geq 0$, from \eqref{eq:SDE} we deduce that for almost all $t\geq 0$:
	\begin{align*}
	\partial_t h(t)
	&= 2\lambda \EE \int_0^{2\pi} \frac{d\theta}{2\pi} \int_\RR f_t(dz) [ |Z_t\cos\theta - z\sin\theta|^r - |Z_t|^r ] \\
	& \quad {} + \mu \EE \int_\RR [|w|^r - |Z_t|^r]\gamma(dw)
	+ rA \EE |Z_t|^r \\
	& \leq A_r h(t)  + B_r \EE\int_\RR |Z_t|^{r-1} |z| f_t(dz) + \tilde{B}_r,
	\end{align*}
	where $A_r = -2\lambda C_r - \mu + rA$, and $B_r,\tilde{B}_r$ are some constants depending on $r$. By Jensen's inequality, we have $\EE|Z_t|^{r-1} \leq h(t)^{1-1/r}$, and also $\int |z| f_t(dz) \leq \max\{E, \int f_0(dv) v^2\}^{1/2}$, because of equation \eqref{eq:energy.t}. Thus, we obtain
	\[
	A_r h(t) +  \tilde{B}_r
	\leq \partial_t h(t)
	\leq A_r h(t) + \max\left\{E, \int f_0(dv) v^2 \right\}^{1/2}B_r h(t)^{1-1/r} + \tilde{B}_r,
	\]
	where the first inequality is deduced similarly, using that $a^r + b^r \leq (a+b)^r$ for all $a,b\geq 0$. From these inequalities, the conclusion follows.
\end{proof}

We introduce the Boltzmann process \eqref{eq:SDE} since it allows us to use coupling arguments for obtaining convergence results in 2-Wasserstein distance, such as contraction and propagation of chaos. Regarding the former, we have:

\begin{lem}[contraction]
	\label{lem:contraction}
	Let $(f_t)_{t\geq 0}$ and $(\tilde{f}_t)_{t\geq 0}$ be two weak solutions to \eqref{eq:PDE} starting from possibly different initial conditions $f_0$ and $\tilde{f}_0$. Then
	\[
	W_2(f_t,\tilde{f}_t) \leq W_2(f_0,\tilde{f}_0)
	        e^{- \frac{\mu T}{2E} t }.
	\]
\end{lem}

\begin{proof}
	For all $t\geq 0$, let $F_t$ be an optimal coupling between $f_t$ and $\tilde{f}_t$, that is, $F_t$ is a probability measure on $\RR\times\RR$ such that $\int (z-\tilde{z})^2 F_t(dz,d\tilde{z}) = W_2(f_t,\tilde{f}_t)^2$. Let $\mathcal{S}(dt,d\theta,dz,d\tilde{z})$ be a Poisson point measure on $[0,\infty)\times[0,2\pi)\times\RR\times\RR$ with intensity $2\lambda dt \frac{d\theta}{2\pi}F_t(dz,d\tilde{z})$, and define $\P(dt,d\theta,dz) = \mathcal{S}(dt,d\theta,dz,\RR)$ and $\tilde{\P}(dt,d\theta,d\tilde{z}) = \mathcal{S}(dt,d\theta,\RR,d\tilde{z})$. In words, $\P$ and $\tilde{\P}$ are Poisson point measures, with intensities $2\lambda dt \frac{d\theta}{2\pi} f_t(dz)$ and $2\lambda dt \frac{d\theta}{2\pi} \tilde{f}_t(d\tilde{z})$ respectively, which have the same atoms in the $t$ and $\theta$ variables, and which are optimally-coupled realizations of $f_t$ and $\tilde{f}_t$ on the $z$ and $\tilde{z}$ variables. Also, let $\Q(dt,dw)$ be a Poisson point measure with intensity $\mu dt \gamma(dw)$ that is independent of $\mathcal{S}$, and set $\tilde{\Q} = \Q$. Let also $(Z_0, \tilde{Z}_0)$ be a realization of $F_0$, independent of everything else; in particular we have $\EE[(Z_0-\tilde{Z}_0)^2] = W_2(f_0,\tilde{f}_0)^2$.
		
	Let $Z_t$ and $\tilde{Z}_t$ be the solutions to the SDE \eqref{eq:SDE} with respect to $(\P,\Q)$ and $(\tilde{\P},\tilde{\Q})$, respectively, thus $\Law(Z_t) = f_t$ and $\Law(\tilde{Z}_t) = \tilde{f}_t$. Consequently, we have $W_2(f_t,\tilde{f}_t)^2 \leq \EE[(Z_t - \tilde{Z}_t)^2] =: h(t)$.
	Using It\^o calculus, we have:
	\begin{align*}
	    \partial_t h(t)
	    &= 2\lambda \EE \int_0^{2\pi} \frac{d\theta}{2\pi} \int_{\RR\times\RR} F_t(dz,d\tilde{z}) [(Z_t\cos\theta - z\sin\theta - \tilde{Z}_t\cos\theta + \tilde{z}\sin\theta)^2 - (Z_t-\tilde{Z}_t)^2] \\
	    & \quad {} + \mu \EE \int_\RR [(w-w)^2 - (Z_t-\tilde{Z}_t)^2] \gamma(dw)
	    + 2\mu \frac{E-T}{2E} \EE (Z_t-\tilde{Z}_t)^2 \\
	    &= 2\lambda \EE \int_0^{2\pi} \frac{d\theta}{2\pi} \int_{\RR\times\RR} F_t(dz,d\tilde{z}) [(\cos^2\theta - 1)(Z_t-\tilde{Z}_t)^2 + (z-\tilde{z})^2\sin^2\theta ]
	    - \frac{\mu T}{E} h(t),
	\end{align*}
	where in the last step the cross term vanished because $\int_0^{2\pi} \cos\theta\sin\theta d\theta = 0$. Since $\int(z-\tilde{z})^2 F_t(dz,d\tilde{z}) = W_2(f_t,\tilde{f}_t)^2 \leq h(t)$, the integral in the last line is bounded above by 0. We thus obtain $\partial_t h(t) \leq - \frac{\mu T}{E} h(t)$, which yields the result.
\end{proof}

\begin{rmk}
	We can also consider the following parametrization of the Kac collisions, introduced by Hauray \cite{hauray2016}: when particles with velocities $v$ and $v_*$ interact, the new velocity for $v$ is $v' = \sqrt{v^2+v_*^2} \cos \theta$. Using this parametrization and the case $\int v^2 f_0(dv) = E = \int v^2 \tilde{f}_0(dv)$, a similar coupling argument yields the following contraction result in $W_4$:
	\begin{align*}
	W_4(f_t,\tilde{f}_t)^4
	&\leq \EE[(Z_0-\tilde{Z}_0)^4] \exp\left( -\left[2\lambda + \frac{2\mu T}{E}-\mu \right] t \right) \\
	& \quad {} + \EE[(Z_0^2-\tilde{Z}_0^2)^2] \exp\left( -\left[\frac{\lambda}{2} + \frac{2\mu T}{E}-\mu \right] t \right).
	\end{align*}
\end{rmk}

Now we show the existence of a steady state $f_\infty$ for \eqref{eq:PDE} and its main properties. Thanks to the contraction property of the solutions in Wasserstein metric (Lemma \ref{lem:contraction}), we can prove:

\begin{lem}[equilibration]
\label{lem:equilibrium}
There exists a probability measure $f_\infty$ on $\RR$ with energy $E$, such that $\lim_{t\to \infty} f_t = f_\infty$ weakly, for any $(f_t)_{t\geq 0}$ weak solution of \eqref{eq:PDE} having bounded initial energy. Moreover,
\[
    W_2(f_t,f_\infty)
    \leq W_2(f_0,f_\infty) e^{-\frac{\mu T}{2E} t},
\]
$f_\infty$ is the unique stationary weak solution of \eqref{eq:PDE}, and is also the unique solution to
\begin{equation}
\label{eq:PDEstationary}
- A \int_\RR \phi'(v) v f_\infty(dv)
= \int_\RR \phi(v) \{ 2\lambda(B_2[f_\infty,f_\infty]-f_\infty)
    +\mu (\gamma(v) - f_\infty)\}dv
\end{equation}
for every bounded test function $\phi$ on $\RR$ with bounded and continuous first derivatives.
\end{lem}

\begin{proof}
For any $t,s\geq 0$, note that $f_{t+s}$ can be seen as the weak solution to \eqref{eq:PDE} at time $t$, starting from $f_s$. Hence, using Lemma \ref{lem:contraction} and the fact that the energy is always bounded by $\max\{ E, \int f_0(dv) v^2\}$, we have
\begin{equation}
\label{eq:W2fts}
    W_2(f_t,f_{t+s})
    \leq  W_2(f_0,f_s) e^{- \frac{\mu T}{2E} t}
    \leq 2 \max\left\{ E, \int v^2 f_0(dv) \right\}^{1/2} e^{-\frac{\mu T}{2E} t}.
\end{equation}
This shows that for any sequence $(t_n)_{n\in\NN}$ with $t_n \to \infty$, the collection $(f_{t_n})_{n\in\NN}$ is a Cauchy sequence with respect to $W_2$, thus it converges to some limit distribution $f_\infty$ (see \cite[Theorem 6.18]{villani2009} for completeness in $W_2$), which is the same for every such $(t_n)_{n\in\NN}$. Thus $(f_t)_{t\geq 0}$ itself converges to $f_\infty$ in $W_2$. Similarly, if $(\tilde{f}_t)_{t\geq 0}$ is the weak solution to \eqref{eq:PDE} starting from $\tilde{f}_0$, then it converges in $W_2$ to some $\tilde{f}_\infty$. Thus, for all $t\geq 0$ we have
\[
W_2(f_\infty, \tilde{f}_\infty)
\leq W_2(f_\infty, f_t)
    + W_2(f_t, \tilde{f}_t)
    + W_2(\tilde{f}_t, \tilde{f}_\infty),
\]
and letting $t\to\infty$ gives $W_2(f_\infty, \tilde{f}_\infty) = 0$, again thanks to Lemma \ref{lem:contraction}. This shows that the limit is the same for any weak solution of \eqref{eq:PDE}. Since the convergence is in $W_2$, we have $\int v^2 f_\infty(dv) = \lim_t \int v^2 f_t(dv) = E$. Letting $s\to\infty$ in \eqref{eq:W2fts} gives $W_2(f_t,f_\infty) \leq W_2(f_0,f_\infty) e^{-\frac{\mu T}{2E} t}$. Moreover, taking $f_0 = f_\infty$, the last inequality implies that $W_2(f_t,f_\infty) = 0$, thus $f_t = f_\infty$ for all $t\geq 0$. That is, $f_\infty$ is a stationary solution, and uniqueness is straightforward.

Now we prove that $f_\infty$ satisfies \eqref{eq:PDEstationary}. By our definition of weak solution, we know that $(g_t)_{t\geq 0}$ given by $\int \phi(v) g_t(dv) = \int \phi(e^{-At}v) f_\infty(dv)$ satisfies \eqref{eq:PDEweak}. Noting that
\begin{align*}
    \int_\RR \phi(v) B_1[t](dv)
    &= \int_\RR \phi(e^{-At}v) \gamma(v) dv
    \quad \text{and} \\
    \int_\RR \phi(v) B_2[g_t,g_t](dv)
    &= \int_\RR \phi(e^{-At}v) B_2[f_\infty,f_\infty](dv),
\end{align*}
we thus obtain for every bounded function $\phi$ with continuous and bounded first derivatives:
\begin{align*}
\int_\RR \phi(e^{-At}v) &\{ 2\lambda(B_2[f_\infty,f_\infty] - f_\infty) + \mu(\gamma(v)-f_\infty) \}(dv) \\
&= \partial_t \int_\RR \phi(e^{-At} v) f_\infty(dv) 
= -A e^{-At} \int_\RR \phi'(e^{-At}v) v f_\infty(dv).
\end{align*}
Evaluating at $t=0$ shows that $f_\infty$ solves \eqref{eq:PDEstationary}. Now we prove that $f_\infty$ is the only solution: given $f\in \P(\RR)$ satisfying \eqref{eq:PDEstationary}, a similar computation shows that $(g_t)_{t\geq 0}$ defined by $\int \phi(v) g_t(dv) = \int \phi(e^{-At}v) f(dv)$ satisfies \eqref{eq:PDEweak} integrated against test functions $\phi$ which are bounded and with bounded first derivatives. By a density argument, it follows that $(g_t)_{t\geq 0}$ itself satisfies \eqref{eq:PDEweak}, thus $f$ is a stationary weak solution of \eqref{eq:PDE}. By the uniqueness of the stationary distribution, we deduce that $f = f_\infty$.
\end{proof}

We now prepare for the proof of Theorem \ref{thm:continuity}. Note that the case $\lambda=0$ in equation \eqref{eq:PDE} can be solved explicitly by the method of characteristics. We summarize the result in the following lemma, which we state without proof.

\begin{lem}\label{lem:8} Let $\lambda=0$ and $f_0 \in L^1(\RR)$. Then the weak solution $(f_t)_{t\geq 0}$ to \eqref{eq:PDE} is given by
\[
    f_t(v) = e^{-\mu t}[ e^{-A t} f_0(e^{-A t}v)] + \int_0^t \mu e^{-\mu s} \left( e^{-A s} \gamma(e^{-A s} v) \right) ds.
\]
Also, $f_\infty$ is given by
\begin{equation}\label{eq:lambda0finfty}
    f_\infty(v) = \int_0^\infty e^{-(1+\frac{A}
    {\mu})s}  \gamma(e^{-\frac{A}{\mu} s} v) ds.
\end{equation}
\end{lem}

\begin{rmk}
When $\lambda=0$, the measure $f_\infty$ given by \eqref{eq:lambda0finfty} can be seen as the law of $X e^{\frac{E-T}{2E} \tau}$, where $X$ and $\tau$ are independent random variables, with $X \sim \gamma$ and $\tau \sim \exp(1)$.
\end{rmk}

In Theorem \ref{thm:continuity}, we will use the Fourier transform to prove the smoothness properties of $f_\infty$. We use the convention
\[
    \hat{\nu}(\xi) 
    = \int_\RR e^{-2\pi i v\xi} \nu(dv)
\]
for any finite measure $\nu$. Note that for any such $\nu_1$ and $\nu_2$ we have $\widehat{B_2[\nu_1, \nu_2]} = \hat{B}_2[\hat{\nu}_1,\hat{\nu}_2]$, where
\begin{equation}
    \label{eq:FourierB2}
    \hat{B}_2[z,w](\xi) := \int_0^{2\pi} z(\xi\cos\theta) w(\xi\sin\theta) \frac{d\theta}{2\pi}.
\end{equation}
Call $y(\xi) = \hat{f}_\infty(\xi)$. We first provide some properties of $y(\xi)$ in the following two lemmas. The proofs are given in sections \ref{sec:Fourier} and \ref{sec:techincal}. For convenience, we introduce the following terminology: we say that a function $\phi:\RR\to \RR \cup \{\infty\}$ satisfies property (P) if it is non-negative, even, and non-increasing on $[0,\infty)$.

\begin{lem}\label{lem:Fourier} For the Fourier transform $y(\xi)$, we have:
\begin{enumerate}
    \item $y \in C^2(\RR)$ and satisfies for all $\xi\in\RR$:
    \begin{equation}
    \label{eq:FourierInfinity}
    - A \xi \frac{dy}{d\xi}(\xi) = 2\lambda \left( \int_0^{2\pi} y(\xi\cos\theta) y(\xi\sin\theta)\frac{d\theta}{2\pi} - y(\xi) \right) + \mu (\hat{\gamma}(\xi) - y(\xi)).
    \end{equation}
\item $y(\xi)$ satisfies (P).
\item For any $k=0,1,2, \dots$, we have $\left(\xi \frac{d}{d\xi}\right)^k y \in C(\RR)$ and $\displaystyle \lim_{\vert \xi\vert \rightarrow\infty}\left(\xi \frac{d}{d\xi}\right)^k y = 0$.
     In particular, there is a constant $C_k$ for which we have
     \[
     \left| \frac{d^k y}{d\xi^k}(\xi) \right|
     \leq C_k |\xi|^{-k}.
     \]
\end{enumerate}
\end{lem}

\begin{lem}\label{lem:GaussTrick}
If $z(\xi): [0, \infty)\rightarrow [0,1]$ is continuous and satisfies 
\begin{equation}\label{eq:intineq}
    z(\xi) \leq r \hat{\gamma}(\xi) + (1-r) \int_0^{\frac{\pi}{2}} z(\xi \cos\theta)z(\xi\sin\theta)\frac{d\theta}{\pi/2},
\end{equation}
for some $r \in (0,1)$, then
\begin{equation}\label{eq:zleqgamma}
    z(\xi)\leq \hat{\gamma}(\xi) \qquad \forall \xi.
\end{equation}
When $r=0$, equation \eqref{eq:zleqgamma} still holds if we assume in addition that $\lim_{\xi \downarrow 0} (z(\xi)- 1)\xi^{-2} = -2 \pi^2 E$.
\end{lem}

\begin{proof}[Proof of Theorem \ref{thm:continuity}]
The existence $f_\infty$, the convergence of $f_t$ in $W_2$ and the fact that it is the unique stationary weak solution of \eqref{eq:PDE}, were proven in Lemma \ref{lem:equilibrium}.
We now prove that $f_\infty$ has a density. For this, we write $f_\infty$ as $\lim_{t\rightarrow \infty} f_t$ where $f_0:= \gamma$. Since $f_0$ satisfies property (P), so does $f_t$ for all $t\geq 0$, thanks to Remark \ref{rmk:property_P}; thus $f_t(v) \leq \frac{1}{2\vert v\vert}$ for all $v\neq 0$ due to the inequality  $1/2 \geq \int_0^v f_t(x) dx \geq v f_t(v)$ which holds for all $v\geq 0$. Using Lemma \ref{lem:equilibrium} to pass to the weak-limit of infinite $t$, we have that for every bounded and continuous function $\phi$ we have $\int \phi(v) f_\infty(dv) \leq \int \phi(v) \frac{dv}{2|v|}$.
This, together with an approximating argument using continuous and bounded functions that $\phi$ vanish around 0, implies that $f_\infty$ necessarily has the form
\begin{equation}\label{eq:delta}
f_\infty = r \delta_0 + (1-r) \rho,
\end{equation}
where $\delta_0$ is the Dirac mass at 0, $\rho(v)$ is some density function satisfying (P) and the estimate $\rho(v) \leq \frac{1}{2|v|}$ for all $v$, and $r\in[0,1)$ is some constant.
Our next task is to show that $r=0$. Equation \eqref{eq:delta} implies that there is an $L^1(\RR)$ probability density $\rho_1$ such that the following holds
\begin{equation}\label{eq:delta2}
    B_2[f_\infty, f_\infty] = r^2 \delta_0 + (1- r^2) \rho_1.
\end{equation}
We substitute equations \eqref{eq:delta} and \eqref{eq:delta2} into equation \eqref{eq:PDEstationary} and let $\phi(v)= e^{-\eta v^2}$. This gives
\begin{align*}
&2 A (1-r) \int_\RR \eta v^2 e^{-\eta v^2} \rho(v)dv \\
&= 2\lambda (r^2- r) - \mu r
+ \int_\RR e^{-\eta v^2}(2\lambda((1-r^2) \rho_1(v) - (1-r) \rho(v)) + \mu \gamma(v)) dv.
\end{align*}
Letting $\eta\rightarrow \infty$ and applying the dominated convergence implies that $2\lambda (r^2- r) - \mu r$ must be $0$. The only solution $r$ in $[0,1]$ compatible with having energy $E > 0$ is $r=0$. This shows that $f_\infty$ has a density that satisfies (P), thus proving (i).

Now we show (ii). The forward implication follows directly from choosing $f_\infty$ as the initial condition and using Lemma \ref{lem:moments}. For the converse: assume \eqref{eq:unif-r-moment} holds. Take $f_0$ as any density with finite $r^{\text{th}}$ moment and let $\phi_M(x) = \max(M, \vert x\vert^r)$. Since $f_t \to f_\infty$ weakly, Lemma \ref{lem:moments} implies that we have
\[
\int_\RR \phi_M(v) f_\infty(v) dv
= \lim_{t\rightarrow \infty}\int_\RR \phi_M(v) f_t(v) dv
\leq \sup_{t\geq 0} \int_\RR \vert v\vert^r f_t(v)dv
< \infty.
\]
The conclusion follows from the monotone convergence theorem.

Now, we prove (iii). To show that $f_\infty \in C^k(\RR\backslash\{0\})$ for any $k$, we use the Fourier transform. Fix $k$, and let $v>0$. We will express $f_\infty(v)$ in terms of its Fourier transform:  $y(\xi)$. Since $f_\infty \in L^1(\RR)$, we have that
\[ 
f_\infty(v) = \lim_{V\rightarrow \infty} \int_{-V}^V \left(1 - \frac{\vert \xi \vert}{V}\right) y(\xi) e^{2\pi i v.\xi} d\xi
\]
Here the limit is taken in $L^1(\RR)$. We can divide this integral into an integral over $[-1,1]$, which is a smooth function in $v$, and an integral over $\{ 1 \leq \vert \xi \vert \leq V\}$. On this interval, we also note that the integrand can be expressed as 
\begin{equation}
\label{eq:FourierLebesgue}
c y(\xi)\left(1 - \frac{\vert \xi \vert}{V}\right) \frac{d^k}{d\xi^k}\left( e^{2\pi i v} v^{-k}\right).
\end{equation}
After performing $k$ integrations-by-part and then taking the limit as $V\rightarrow \infty$, we see that the contribution of $\{ \vert \xi \vert \geq 1\}$ is   
\[
\int_{\{\vert \xi\vert \geq 1\}} c e^{-2\pi i v} v^{-k}  \frac{d^k}{d\xi^k} y(\xi) d\xi
\]
plus boundary terms at $\pm 1$. The boundary terms at $\pm 1$ are smooth in $v$ so long as $v>0$, while property (iii) in Lemma \ref{lem:Fourier} implies that the boundary terms at $\pm V$ vanish as $V\rightarrow \infty$. The integral in \eqref{eq:FourierLebesgue} is in the sense of Lebesgue; and is in $C^{k-2}$ because $y^{(k)}(\xi)$ decays like $\xi^{-k}$. In particular, $f_\infty$ is in $C^{k-2}$ away from $v=0$. Similarly one can show that $B_2[f_\infty, f_\infty] \in C^{k-2}(\RR\backslash\{0\})$, justifying the pointwise equality in \eqref{eq:PDEstationary-strong}.

Finally, we prove (iv), concerning the boundedness and smoothness properties of $f_\infty$ at the origin. We start by proving the last claim in (iv) concerning the analyticity of $y$. When $E \geq T$ ($A \geq 0$), we have $-A \xi \frac{\partial y}{\partial \xi} \geq 0$ on $[0, \infty)$. Thus, on $[0, \infty)$ we have:
\[
    y(\xi) \leq \frac{2\lambda}{2\lambda+\mu} \hat{\gamma}(\xi) +  \frac{\mu}{2\lambda+\mu}\int_0^{\pi/4} y(\xi \cos\theta) y(\xi \sin\theta)\, \frac{d\theta}{\pi/4}.
\]
Thus, by Lemma \ref{lem:GaussTrick}, we have $\displaystyle  y(\xi) \leq \hat{\gamma}(\xi)$. Therefore $y(\xi)$ is in $L^1(\RR)$, making $f_\infty$ analytic.
Therefore, we need only consider the case when $E<T$ ($A<0$). Note that \eqref{eq:Cp-Condition} is equivalent to $D > - A (p+ 1)$.
We rewrite equation \eqref{eq:FourierInfinity} in the following form. Recall that $D=2\lambda+\mu$.
\begin{equation}\label{eq:Fourier:aux}
    \left( \xi^{D/\vert A\vert } y(\xi)\right)' = \xi^{D/\vert A\vert - 1} \left( \frac{2\lambda}{\vert A\vert} \hat{B}_2(\xi) + \frac{\mu}{\vert A\vert} \hat{\gamma}(\xi) \right)
\end{equation}
Here $\hat{B}_2(x) = \int_{\theta=0}^{2\pi} y(x\cos\theta)y(x\sin\theta)\frac{d\theta}{2\pi}$ as in Lemma \ref{lem:Fourier}. Note that all the constants are positive. We integrate both sides of \eqref{eq:Fourier:aux} on $[0,\xi]$ and multiply the resulting equation by $\xi^{p-D/\vert A\vert }$. This leads to the equation:
\begin{equation*}
    y(\xi) \xi^p = \xi^{p-D/\vert A\vert }\int_0^\xi x^{D/\vert A\vert  - 1} \left( \frac{2\lambda}{\vert A\vert} \hat{B}_2(x) + \frac{\mu}{\vert A\vert} \hat{\gamma}(x) \right) dx
\end{equation*}
Hence, we can show that $\int_\RR y(\xi)\vert \xi\vert^B d\xi$  is infinite when $B$ is large enough as follows.
\begin{eqnarray}\label{eq:typofixed}
 \int_0^\infty y(\xi) \xi^p d\xi & = & \int_{\xi=0}^\infty \xi^{p-D/\vert A\vert }\int_0^\xi x^{D/\vert A\vert  - 1} \left( \frac{2\lambda}{\vert A\vert } \hat{B}_2(x) + \frac{\mu}{\vert A\vert } \hat{\gamma}(x) \right) dx \, d\xi \nonumber \\& = & \int_{x=0}^\infty  x^{D/\vert A\vert - 1} \left( \frac{2\lambda}{\vert A\vert } \hat{B}_2(x) + \frac{\mu}{\vert A\vert } \hat{\gamma}(x) \right) \left( \int_{\xi=x}^\infty \xi^{p-D/\vert A\vert } d\xi \right) dx
\end{eqnarray}
 When $p\geq \frac{D}{\vert A\vert}-1$, this is infinite due to the $\xi$-integral.
  We now show that if $p< \frac{D}{\vert A\vert}-1$, then $\int_\RR y(\xi) \xi^p d\xi $ is finite.
 Consider the integral $\int_0^X \xi^p \hat{B}_2(\xi)d\xi$. As above, we restrict the average over the $\theta$ variable to the interval $[0, \pi/4]$. By using polar coordinates: $(x,z):= (\xi \cos\theta, \xi \sin\theta)$, augmenting the domain to the triangle $\{y \leq x \text{ and } x\leq R \}$, and applying the change of variables $z = x \times t$ we obtain:
 \begin{equation}\label{eq:B2hat_is_small}
  \int_0^X \xi^p \hat{B}_2(\xi)d\xi \leq  \int_{x=0}^X x^p y(x) \left( \frac{4}{\pi} \int_{t=0}^1 (1+t^2)^{\frac{p-1}{2}} y(t x) \,dt \right) dx
 \end{equation}
  Equation \eqref{eq:B2hat_is_small}, together with the fact that $\lim_{\xi\rightarrow \infty} y(\xi)= 0$, implies that 
  \[
        \int_0^X \xi^p \hat{B}_2(\xi) \,d\xi \leq \frac{X_0^{p+1}}{p+1} + \epsilon \int_{X_0}^X \xi^p y(\xi)\, d\xi, 
  \]
 for any $\epsilon$ whenever $X \geq X_0(\epsilon)$ is large enough. Thus, multiplying \eqref{eq:FourierInfinity} by $\xi^p$ and integrating the result on the interval $[0,X]$ (with $X$ large enough) , and performing an integration by parts gives: 
\[ 
-A X^{p+1} y(X) + (2\lambda + \mu - 2\lambda \epsilon + A (p+1)) \int_0^X \xi^p y(\xi) d\xi \leq 2\lambda \frac{X_0^{p+1}}{p+1} + \mu \int_0^\infty \xi^p \hat{\gamma}(\xi) d\xi
\]
Since $A<0$, we can neglect the term: $-A X^{p+1} y(X)$. Also, because
$D- \vert A\vert (p+1) >0$, and w.l.o.g. $\lambda>0$, choosing $\epsilon < (D-|A|(p+1))/2\lambda$ and letting $X \rightarrow \infty$ shows that $\int_0^\infty \xi^p y(\xi) d\xi <\infty$.

\end{proof}

%%%%%%%%%%%%%%%%
% Particle system
%%%%%%%%%%%%%%%%

\section{Particle system}\label{sec:particle-system}

In this section we study the propagation of chaos for the particle system associated to the equation \eqref{eq:PDE} with rescaling factor $\beta_N$ given in \eqref{eq:betaN}. We provide an explicit construction of the particle system using an SDE, following \cite{cortez-fontbona2016}. To this end, for fixed $N\in\NN$, let $\R(dt,d\theta,d\xi,d\zeta)$ be a Poisson point measure on $[0,\infty)\times[0,2\pi)\times [0,N)^2$ with intensity
\[
N\lambda dt \frac{d\theta}{2\pi} \frac{d\xi d\zeta \ind_{\{\ii(\xi) \neq \ii(\zeta)\}}}{N(N-1)}
= \frac{\lambda dt d\theta d\xi d\zeta \ind_{\{\ii(\xi) \neq \ii(\zeta)\}}}{2\pi(N-1)},
\]
where $\ii$ is the function that associates to a variable $\xi\in[0,N)$ the discrete index $\ii(\xi) = \lfloor \xi \rfloor +1 \in \{1,\ldots,N\}$. In words: at rate $N\lambda$, the measure $\R$ selects collision times $t\geq 0$, and for each such time, it independently samples a parameter $\theta$ uniformly at random on $[0,2\pi)$, and a pair $(\xi,\zeta) \in [0,N)^2$ such that $\ii(\xi) \neq \ii(\zeta)$, also uniformly. The pair $(\ii(\xi),\ii(\zeta))$ provides the indices of the particles involved in Kac-type collisions. Let $\Q_1(dt,dw),\ldots,\Q_N(dt,dw)$ be a collection of independent Poisson point measures, also independent of $\R$, each having intensity $\mu dt \gamma(dw)$. Finally, let $\mathbf{V}_0 = (V_0^1,\ldots,V_0^N)$ be an exchangeable collection of random variables, independent of everything else, such that $\EE[(V_0^1)^2] = E$.

The particle system, which we denote by $\mathbf{V}_t = (V_t^1,\ldots,V_t^N)$, is defined as the unique jump-by-jump solution of the SDE
\begin{equation}
    \label{eq:sde_ps}
    \begin{split}
        d\mathbf{V}_t
        &= \int_0^{2\pi} \int_{[0,N)^2} \sum_{i,j=1, i\neq j}^N [ \mathbf{a}_{ij}(\mathbf{V}_{t^\-},\theta) - \mathbf{V}_{t^\-} ] \ind_{\{ \ii(\xi)=i,\ii(\zeta)=j\}} \R(dt,d\theta,d\xi,d\zeta) \\
        &\quad {} + \sum_{i=1}^N \int_\RR [\mathbf{b}_i(\mathbf{V}_{t^\-},w) - \mathbf{V}_{t^\-}] \Q_i(dt,dw).
    \end{split}
\end{equation}
that starts at $\mathbf{V}_0$.
Here, for $\mathbf{v} \in \RR^N$, the vectors $\mathbf{a}_{ij}(\mathbf{v},\theta)\in\RR^N$ and $\mathbf{b}_i(\mathbf{v},w)\in\RR^N$ are defined as
\[
\mathbf{a}_{ij}(\mathbf{v},\theta)^k
= \begin{cases}
v^i \cos\theta - v^j \sin\theta & \text{if $k=i$,} \\
v^i \sin\theta + v^j \cos\theta & \text{if $k=j$,} \\
v^k & \text{otherwise},
\end{cases}
\qquad
\mathbf{b}_i(\mathbf{v},w)^k
= \begin{cases}
\beta_N(w)w & \text{if $i=k$,} \\
\beta_N(w)v^k & \text{otherwise},
\end{cases}
\]
where $\beta_N(w)$ is the rescaling factor given by
\[
    \beta_N(w) = \sqrt{\frac{NE}{NE-E+w^2}}.
\]
We call $\mathbf{V}_t$ the particle system, and denote $f_t^N = \Law(\mathbf{V}_t)$ its distribution on $\RR^N$. For any $i=1,\ldots,N$, from \eqref{eq:sde_ps} it follows that the particle $V_t^i$ satisfies the SDE 
\begin{equation}
\label{eq:sde_vi}
    \begin{split}
        dV_t^i
        &= \int_0^{2\pi} \int_0^N [V_{t^\-}^i \cos\theta - V_{t^\-}^{\ii(\xi)}\sin\theta - V_{t^\-}^i] \P_i(dt,d\theta,d\xi) \\
        & \quad {} + \int_\RR [\beta_N(w)w - V_{t^\-}^i] \Q_i(dt,dw)
        + \sum_{j=1, j\neq i}^N \int_\RR [\beta_N(w)V_{t^\-}^i - V_{t^\-}^i] \Q_j(dt,dw),
    \end{split}
\end{equation}
where $\P_i$ is defined as
\[
\P_i(dt,d\theta,d\xi)
= \R(dt,d\theta,[i-1,i), d\xi)
+ \R(dt,-d\theta, d\xi, [i-1,i)),
\]
and where we have $-d\theta$ to transform $\sin\theta$ into $-\sin\theta$. Clearly, $\P_i$ is a Poisson point measure on $[0,\infty)\times[0,2\pi)\times[0,N)$ with intensity $\frac{2 \lambda dt d\theta d\xi \ind_{\{\ii(\xi)\neq i\}}}{2\pi(N-1)}$.

\begin{lem}[contraction for the p.s.]
\label{lem:contraction_ps}
Let $\tilde{\mathbf{V}}_t$ be the solution to \eqref{eq:sde_ps} starting from a possibly different exchangeable initial condition $\tilde{\mathbf{V}}_0$. Then for all $t\geq 0$,
\[
\EE[ (V_t^1 - \tilde{V}_t^1)^2]
= \EE[ (V_0^1 - \tilde{V}_0^1)^2] e^{-G_N \mu t},
\]
where $G_N := \int \frac{N w^2 \gamma(dw)}{NE-E+w^2} >0$ uniformly in $N$. Denoting $\tilde{f}_t^N = \Law(\tilde{\mathbf{V}}_t)$, this implies that
\[
W_2(f_t^N, \tilde{f}_t^N)
\leq W_2(f_0^N, \tilde{f}_0^N) e^{-\frac{G_N \mu}{2} t}.
\]
\end{lem}

\begin{proof}
Since the pair $(\mathbf{V}_0,\tilde{\mathbf{V}}_0)$ can be chosen as an optimal coupling between $f_0^N$ and $\tilde{f}_0^N$, the second claim is a direct consequence of the first, because $ W_2(f_0^N, \tilde{f}_0^N)^2 \leq \EE[(V_t^1-\tilde{V}_t^1)^2]$ by exchangeability. Let us prove the first claim: for $h(t) = \EE[(V_t^1-\tilde{V}_t^1)^2]$, using exchangeability again, we have from \eqref{eq:sde_vi}:
\begin{align*}
    \partial_t h(t)
    &= 2\lambda \EE \int_0^{2\pi} \frac{d\theta}{2\pi} [( (V_t^1-\tilde{V}_t^1) \cos\theta - (V_t^2 - \tilde{V}_t^2) \sin\theta)^2 - (V_t^1-\tilde{V}_t^1)^2] \\
    &\quad {} + \mu \EE \int_\RR [(\beta_N(w)w - \beta_N(w)w)^2 - (V_t^1-\tilde{V}_t^1)^2] \gamma(dw)\\
    &\quad {} + (N-1)\mu\EE \int_\RR [(\beta_N(w)V_t^1 - \beta_N(w)\tilde{V}_t^1)^2 - (V_t^1-\tilde{V}_t^1)^2] \gamma(dw)\\
    & = \left(-1  + (N-1) \int_\RR [\beta_N(w)^2 -1] \gamma(dw)\right) \mu h(t),
\end{align*}
where the first term vanished because $\int_0^{2\pi} \cos\theta \sin\theta d\theta = 0$. A straightforward computation shows that the constant in front of $\mu h(t)$ in the last line equals  $-G_N$, which proves the first claim. Finally, note that $\frac{Nw^2}{NE-E+w^2}$ is bounded below by $\frac{Nw^2}{NE+w^2}$, which increases with $N$; thus, $G_N \geq \int \frac{2w^2 \gamma(dw)}{2E+w^2} >0$ uniformly in $N$.
\end{proof}

In the usual Kac particle system (without thermostat and rescaling) the initial energy $\sum_i (V_0^i)^2$ is a.s.\ preserved. Although this is no longer the case for our p.s., we still have preservation of the \emph{expected} energy:

\begin{lem}\label{lem:EV-squared} $\EE[(V_t^1)^2] = E$ for all $t\geq 0$.
\end{lem}

\begin{proof}
For $h(t) := \EE[(V_t^1)^2]$, a similar argument as in the proof of Lemma \ref{lem:contraction_ps} gives $\partial_t h(t) = \mu E G_N - \mu G_N h(t)$, which implies that $h(t) = E$ for all $t\geq 0$, because $h(0) = \EE[(V_0)^2] = E$.
\end{proof}

The above results imply the existence of an equilibrium distribution $f_\infty^N$ for the particle system, which is $f_\infty$-chaotic. This is the content of the following:

\begin{lem}[equilibration for the particle system]
\label{lem:equilibration_ps}
For each $N \geq 2$,  there is a probability measure $f_\infty^N$, depending on the parameters $\lambda, \mu, E, $ and $T$, such that for any exchangeable initial condition, the distribution $f^N_t = \Law(\mathbf{V}_t)$ satisfies
\begin{equation}
\label{eq:equilibration_ps}
W_2(f_t^N, f_\infty^N)
\leq W_2(f_0^N, f_\infty^N) e^{-\frac{G_N \mu}{2} t}
\qquad \forall t \geq 0.
\end{equation}
Moreover, if $r>2$ is chosen so that condition \eqref{eq:unif-r-moment} holds (recall that it always exists), then there exists a constant $C$ depending on $r$, $\lambda$, $\mu$, $E$, and $T$, such that for $\mathbf{U}$ with distribution $f_\infty^N$, we have the following chaoticity estimate of $f_\infty^N$ with respect to $f_\infty$:
\begin{equation}
\label{eq:equilibration_ps_chaos}
    \EE[W_2(\bar{\mathbf{U}}, f_\infty)^2]
    \leq C \begin{cases}
        N^{-1/3}, & r>4, \\
        N^{-1/3}(\log N)^{2/3}, & r=4, \\
        N^{- \frac{r-2}{2(r-1)}}, &  r \in (2,4).
    \end{cases}
\end{equation}
\end{lem}

\begin{proof}
Equation \eqref{eq:equilibration_ps} follows from Lemma \ref{lem:contraction_ps} with the same argument as in the proof of Lemma \ref{lem:equilibrium}, so we omit the details. We now prove \eqref{eq:equilibration_ps_chaos}. Take $f_0^N = f_\infty^N$ and $f_0 = f_\infty$, thus $f_t^N = f_\infty^N$ and $f_t = f_\infty$ for all $t\geq 0$. Therefore, $\EE[W_2(\bar{\mathbf{U}}, f_\infty)^2] = \EE[W_2(\bar{\mathbf{V}}_t, f_t)^2]$ for all $t\geq 0$. We also know from Theorem \ref{thm:continuity}-(ii) that $f_\infty$ has finite $r^{th}$-moment. Hence, applying Theorem \ref{thm:UPoC} (to be proven shortly), \eqref{eq:UPoC} holds. Equation \eqref{eq:equilibration_ps_chaos} follows by letting $t\to\infty$.
\end{proof}

We now prepare the grounds for proving our uniform propagation of chaos result, Theorem \ref{thm:UPoC}. The proof is based on a coupling argument. The main idea, introduced in \cite{cortez-fontbona2016}, is to define a collection $\mathbf{Z}_t = (Z_t^1,\ldots,Z_t^N)$ of Boltzmann processes in such a way that, for each $i=1,\ldots,N$, the process $Z_t^i$ remains as close as possible to the particle $V_t^i$. The construction goes as follows. First, consider a random vector $\mathbf{Z}_0 = (Z_0^1,\ldots,Z_0^N)$ with distribution $f_0^{\otimes N}$, independent of $\R$ and of the $\Q_i$'s, optimally coupled to $\mathbf{V}_0$, that is,
\begin{equation}
\label{eq:W2f0N}
    \EE\left[ \frac{1}{N} \sum_{i=1}^N (V_0^i - Z_0^i)^2 \right]
    = W_2(f_0^N, f_0^{\otimes N})^2.
\end{equation}
Now, mimicking \eqref{eq:SDE} and \eqref{eq:sde_vi}, we define $Z_t^i$ as the solution, starting from $Z_0^i$, to the SDE
\begin{equation}
\label{eq:sde_zi}
    \begin{split}
        dZ_t^i
        &= \int_0^{2\pi} \int_0^N [Z_{t^\-}^i \cos\theta - F_t^i(\mathbf{Z}_{t^\-},\xi)\sin\theta - Z_{t^\-}^i] \P_i(dt,d\theta,d\xi) \\
        & \quad {} + \int_\RR [w - Z_{t^\-}^i] \Q_i(dt,dw)
        + A Z_t^i dt,
    \end{split}
\end{equation}
using the same Poisson point measures $\P_i$ and $\Q_i$ as in \eqref{eq:sde_vi}. Here, $F_t^i : \RR^N \times [0,N) \to \RR$ is a measurable function with the following property: for each $\mathbf{z} \in \RR^N$ and any uniformly distributed random variable $\xi$ on the set $[0,N) \backslash [i-1,i)$, the pair $(F_t^i(\mathbf{z},\xi), z^{\ii(\xi)})$ is a realization of the optimal coupling between $f_t$ and the empirical measure $\bar{\mathbf{z}}^i := \frac{1}{N-1}\sum_{j\neq i} \delta_{z^j}$. That is, for all $\mathbf{z} \in \RR^N$,
\begin{equation}
    \label{eq:EW2ftZi}
    \int_0^N \left(F_t^i(\mathbf{z},\xi) - z^{\ii(\xi)}\right)^2 \frac{d\xi \ind_{\{\ii(\xi)\neq i\}}}{N-1}
    = W_2(f_t, \bar{\mathbf{z}}^i)^2.
\end{equation}
The proof of existence of such a function $F_t^i$ can be found in \cite[Lemma 3]{cortez-fontbona2016}.

Comparing \eqref{eq:SDE} and \eqref{eq:sde_zi}, it is clear that each $Z_t^i$ is indeed a Boltzmann process, thus $\Law(Z_t^i) = f_t$ for all $i=1,\ldots,N$. In words, the dynamics of $Z_t^i$ mimics as closely as possible the one of $V_t^i$: it uses the same jump times, the same collision parameters $\theta$, and the same samples of the thermostat. Moreover, for $t\geq 0$ fixed, the expression $V_{t^\-}^{\ii(\xi)}$ in \eqref{eq:sde_vi} corresponds to a $\xi$-realization of the (random) empirical measure $\frac{1}{N-1} \sum_{j\neq i} \delta_{V_{t^\-}^j}$, while $F_t^i(\mathbf{Z}_{t^\-},\xi)$ in \eqref{eq:sde_zi} is a $\xi$-realization of $f_t$ that is optimally coupled to the empirical measure $\frac{1}{N-1} \sum_{j\neq i} \delta_{Z_{t^\-}^j}$. However, regarding the rescaling mechanism, we note that, while in the system an interaction of any particle $j\neq i$ against the thermostat (slightly) modifies the velocity $V_t^i$ (last term in \eqref{eq:sde_vi}), for the process $Z_t^i$ this mechanism takes the form of the drift term $AZ_t^i dt$.

Because the vector $\mathbf{Z}_t = (Z_t^1,\ldots,Z_t^N)$ uses the same randomness as that of the particle system, the processes $Z_t^i$ and $Z_t^j$ have a simultaneous jump whenever a Kac-type interaction between $V_t^i$ and $V_t^j$ takes place. Consequently, the processes $Z_t^1,\ldots,Z_t^N$ are \emph{not independent}. Thus, to prove Theorem \ref{thm:UPoC}, we need to show that this dependence vanishes as $N\to\infty$. We do this in a quantitative way and uniformly in time in the following lemma. 

\begin{lem}[decoupling]
\label{lem:decoupling}
There exists a constant $C<\infty$ depending only on $E$, $T$, $\lambda$ and $\mu$, such that for all fixed $k\in\NN$ we have for all $t\geq 0$:
\[
W_2\left(f_t^{\otimes k}, \Law(Z_t^1,\ldots,Z_t^k) \right)^2
\leq \frac{Ck}{N}.
\]
\end{lem}

\begin{proof}
We follow the proof of \cite[Lemma 6]{cortez-fontbona2016}, which relies on a coupling argument; see also \cite[Lemma 3]{cortez2016}. For fixed $k\in\NN$, the plan is to define a collection $\tilde{Z}_t^1,\ldots,\tilde{Z}_t^k$ of \emph{independent} Boltzmann processes that they stay close to $Z_t^1,\ldots,Z_t^k$ on expectation. For $i\in\{1,\ldots,k\}$, the jumps of the process $\tilde{Z}_t^i$ will use the same randomness source that $Z_t^i$ uses, except when there is a simultaneous Kac-type jump with some $Z_t^j$, $j\in\{1,\ldots,k\}$, $j\neq i$; in that case, either $\tilde{Z}_t^i$ or $\tilde{Z}_t^j$ will not jump at that instant. To compensate for the missing jumps, we will use an additional, independent source of randomness to generate new jumps. For $k \ll N$, this kind of jumps occur much less frequently, thus this construction will give the desired estimate.

To this end, let $\tilde{\R}$ be an independent copy of the Poisson point measure $\R$, and for $i=1,\ldots,k$, define
\begin{align*}
\tilde{\P}_i(dt,d\theta,d\xi)
&= \R(dt,d\theta,[i-1,i), d\xi) \\
& \qquad {} + \R(dt,-d\theta, d\xi, [i-1,i)) \ind_{[k,N)}(\xi) \\
& \qquad {} + \tilde{\R}(dt, -d\theta, d\xi, [i-1,i)) \ind_{[0,k)}(\xi),
\end{align*}
which is a Poisson point measure with intensity $\frac{2 \lambda dt d\theta d\xi \ind_{\{\ii(\xi)\neq i\}}}{2\pi(N-1)}$, just as $\P_i$. Note that the Poisson measures $\tilde{\P}_1,\ldots,\tilde{\P}_k$ are independent by construction. Mimicking \eqref{eq:sde_zi}, we define $\tilde{Z}_t^i$ as the solution, starting from $\tilde{Z}_0^i = Z_0^i$, to the SDE
\begin{equation}
    \begin{split}
    \label{eq:SDE_tilde_zi}
    d\tilde{Z}_t^i
    &= \int_0^{2\pi} \int_0^N [\tilde{Z}_{t^\-}^i \cos\theta - F_t^i(\mathbf{Z}_{t^\-},\xi)\sin\theta - \tilde{Z}_{t^\-}^i] \tilde{\P}_i(dt,d\theta,d\xi) \\
    & \quad {} + \int_\RR [w - \tilde{Z}_{t^\-}^i] \Q_i(dt,dw)
    + A \tilde{Z}_t^i dt.
    \end{split}
\end{equation}
It is clear that $\tilde{Z}_t^1,\ldots,\tilde{Z}_t^k$ is an exchangeable collection of Boltzmann processes. Moreover, using the independence of $\tilde{\P}_1,\ldots,\tilde{\P}_k$ and the fact that $F_t^i(\mathbf{z},\xi)$ has distribution $f_t$ for any $\mathbf{z}\in \RR^N$ and any $\xi$ uniformly distributed on $[0,N)\backslash [i-1,i)$, one can prove that the processes $\tilde{Z}_t^1,\ldots,\tilde{Z}_t^k$ are independent. For a full proof of this fact in a very similar setting, we refer the reader to \cite[Lemma 6]{cortez-fontbona2016}.

Call $h(t) := \EE[(Z_t^1 - \tilde{Z}_t^1)^2]$. By exchangeability, we have
\[
W_2\left(f_t^{\otimes k}, \Law(Z_t^1,\ldots,Z_t^k) \right)^2
\leq \EE\left[ \frac{1}{k} \sum_{i=1}^k (Z_t^i - \tilde{Z}_t^i)^2\right]
= h(t),
\]
thus it suffices to obtain the desired estimate for $h(t)$. From \eqref{eq:sde_zi} and \eqref{eq:SDE_tilde_zi}, using It\^o calculus, we obtain:
\begin{equation}
    \begin{split}
    \label{eq:dht_decoupling}
    \partial_t h(t)
    &= \EE \int_0^{2\pi} \int_0^N \Delta_1 \left[\R(dt,d\theta,[0,1), d\xi) + \R(dt,-d\theta,d\xi,[0,1))\ind_{[k,N)}(\xi)\right] \\
    & \qquad {} + \EE \int_0^{2\pi} \int_0^N \Delta_2  \R(dt,-d\theta,d\xi,[0,1))\ind_{[0,k)}(\xi) \\
    & \qquad {} + \EE \int_0^{2\pi} \int_0^N \Delta_3  \tilde{\R}(dt,-d\theta,d\xi,[0,1))\ind_{[0,k)}(\xi) \\
    & \qquad {} + \EE \int_\RR [(w-w)^2 - (Z_{t^\-}^1-\tilde{Z}_{t^\-}^1)^2] \Q_1(dt,dw)+ 2 A \EE[(Z_t^1 - \bar{Z}_t^1)^2],
    \end{split}
\end{equation}
where $\Delta_1$ corresponds to the increment of $(Z_t^1 - \tilde{Z}_t^1)^2$ when $Z_t^1$ and $\tilde{Z}_t^1$ have a simultaneous Kac-type jump, $\Delta_2$ is the increment when only $Z_t^1$ jumps, and $\Delta_3$ is the increment when only $\tilde{Z}_t^1$ jumps. Thanks to the indicator $\ind_{[0,k)}(\xi)$ and the fact that $\Delta_2$ and $\Delta_3$ involve only second-order products of $f_t$-distributed variables (recall that $\int v^2 f_t(dv) = E$ for all $t\geq 0$), we deduce that the second and third terms in \eqref{eq:dht_decoupling} are bounded above by $\frac{Ck}{N}$. On the other hand, for $\Delta_1$ we have
\begin{align*}
    \Delta_1 &= \left[Z_{t^\-}^1 \cos\theta - F_t^1(\mathbf{Z}_{t^\-},\xi) \sin\theta
    -\tilde{Z}_{t^\-}^1 \cos\theta + F_t^1(\mathbf{Z}_{t^\-},\xi) \sin\theta
    \right]^2
    - (Z_{t^\-}^1 - \tilde{Z}_{t^\-}^1)^2 \\
    &= - (1-\cos^2\theta) (Z_{t^\-}^1 - \tilde{Z}_{t^\-}^1)^2.
\end{align*}
Thus, simply discarding the negative term with the indicator $\ind_{[k,N)}(\xi)$ in \eqref{eq:dht_decoupling}, we deduce that
\begin{align*}
    \partial_t h(t)
    &\leq  -\EE \int_0^{2\pi} \int_1^N (1-\cos^2\theta) (Z_t^1 - \tilde{Z}_t^1)^2 \frac{2 \lambda d\theta d\xi}{2\pi(N-1)}
    + \frac{Ck}{N} - \mu h(t) + 2A h(t) \\
    &= -(\lambda + \mu - 2A) h(t) + \frac{Ck}{N}.
\end{align*}
Recall that $A = \mu \frac{E-T}{2E}$, thus $\partial_t h(t) + (\lambda + \frac{\mu T}{E}) h(t) \leq \frac{Ck}{N}$. Since $h(0) = 0$, the desired bound follows from the last inequality by multiplying by $e^{(\lambda + \frac{\mu T}{E})t}$ and integrating.
\end{proof}

We now want to use the previous lemma to obtain a chaoticity estimate for $\mathbf{Z}_t$ in terms of $\EE W_2(\bar{\mathbf{Z}}_t, f_t)^2$. We will need some previous results. Recall that (see \eqref{eq:epsN}) given a probability measure $\nu$ on $\RR$ and $k\in\NN$, we call $\varepsilon_k(\nu) = \EE W_2(\bar{\mathbf{X}}, \nu)^2$, where $\mathbf{X} = (X^1,\ldots,X^k)$ is a vector of i.i.d.\ $\nu$-distributed random variables, and $\bar{\mathbf{X}} = \frac{1}{k} \sum_{i=1}^k \delta_{X^i}$ is its (random) empirical measure. The best available general estimate for $\varepsilon_k(\nu)$, found in \cite[Theorem 1]{fournier-guillin2013}, asserts that if $\nu$ has finite $r$ moment for some $r>2$, then there exists a constant $C>0$ depending only on $r$ such that for all $k\in\NN$,
\begin{equation}
    \label{eq:eps}
    \varepsilon_k(\nu)
    \leq C \left( \int |x|^r \nu(dx) \right)^{2/r}
    \begin{cases}
        k^{-1/2}, & r>4, \\
        k^{-1/2} \log(k), & r=4, \\
        k^{-(1-\frac{2}{r})}, & r\in(0,2).
    \end{cases}
\end{equation}

\begin{rmk}
\label{rmk:eps}
The case $r=4$ is omitted in the statement of \cite[Theorem 1]{fournier-guillin2013}, but their proof can be easily adapted to include it. We provide the needed adjustments in Section \ref{sec:app_Fournier} of the Appendix.
\end{rmk}

We will also need the following estimate, which is a consequence of \cite[Lemma 7]{cortez-fontbona2016}: for any exchangeable random vector $\mathbf{X}$ on $\RR^N$ and any probability measure $\nu$ on $\RR$, there exists a constant $C>0$ depending only on the second moments of $\nu$ and $X^1$, such that for any $k\leq N$,
\begin{equation}
    \label{eq:W2X}
    \frac{1}{2} \EE [W_2(\bar{\mathbf{X}}, \nu)^2]
    \leq W_2\left(\nu^{\otimes k}, \Law(X^1,\ldots,X^k)\right)^2
        + \varepsilon_k(\nu) + C\frac{k}{N}.
\end{equation}

We can now state and prove the following:

\begin{lem}
\label{lem:decouplingEW2}
Fix $r>2$ such that $\int |v|^r f_0(dv) < \infty$ and condition \eqref{eq:unif-r-moment} is satisfied for this $r$. Then there exists a constant $C$ depending only on $\lambda$, $\mu$, $E$, $T$, $r$, and $\int |v|^r f_0(dv)$, such that for all $t\geq 0$,
\begin{equation}
\label{eq:W_2Ztft}
\EE [W_2(\bar{\mathbf{Z}}_t, f_t)^2]
\leq C
\begin{cases}
        N^{-1/3}, & r>4, \\
        N^{-1/3}(\log N)^{2/3}, & r=4, \\
        N^{- \frac{r-2}{2(r-1)}}, &  r \in (2,4).
    \end{cases}
\end{equation}
Moreover, the same bound holds with $\bar{\mathbf{Z}}_t^i = \frac{1}{N-1}\sum_{j\neq i} \delta_{Z_t^j}$ in place of $\bar{\mathbf{Z}}_t$.
\end{lem}

\begin{proof}
For any $k\leq N$, we obtain from \eqref{eq:W2X} applied to $\nu=f_t$ and $\mathbf{X} = \mathbf{Z}_t$:
\begin{align*}
\frac{1}{2} \EE [W_2(\bar{\mathbf{Z}}_t, f_t)^2]
& \leq W_2\left(f_t^{\otimes k}, \Law(Z_t^1,\ldots,Z_t^k)\right)^2
+ \varepsilon_k(f_t) + C\frac{k}{N} \\
& \leq C\frac{k}{N} + \varepsilon_k(f_t) + C\frac{k}{N},
\end{align*}
where in the last step we used Lemma \ref{lem:decoupling}. Since \eqref{eq:unif-r-moment} holds for this $r>2$, we know from Lemma \ref{lem:moments} that the $r^{\text{th}}$ moment of $f_t$ is bounded uniformly in time. If $r>4$, from \eqref{eq:eps} we obtain $\varepsilon_k(f_t) \leq C k^{-1/2}$ for all $t\geq 0$, and then taking $k \sim N^{2/3}$ gives the result. If $r\in(2,4)$, take $k \sim N^{\frac{r}{2(r-1)}}$. Finally, if $r=4$, take $k \sim (N \log N )^{2/3}$. The estimate for $\bar{\mathbf{Z}}_t^i$ is deduced similarly, taking $\mathbf{X} = (Z_t^j)_{j\neq i}$ in \eqref{eq:W2X}.
\end{proof}

The proof of Theorem \ref{thm:UPoC} requires the following two technical lemmas. The proof of Lemma \ref{lem:app} is provided in the Appendix, while the proof of Lemma \ref{lem:Gronwall} can be found for instance in \cite[Lemma 4.1.8]{ambrosio-gigli-savare2008}.

\begin{lem}\label{lem:app}
Let $\nu$ be any probability measure on $\RR$ with $\int_\RR w^2 \nu(dw) = T$ and $\int_\RR w^4 \nu(dw) =: M_4(\nu)< \infty$. Let $\bar{K}^1_N, \bar{K}^2_N,$ and $\bar{K}^3_N$ be given as follows.
\begin{eqnarray*}
    \bar{K}^1_N &=& (N-1) \int_\RR \nu(dw) (1 - \beta_N(w))^2 \\
    \bar{K}^2_N & = & \int_\RR \nu(dw) (\beta_N(w)-1)^2 w^2\\ 
    \bar{K}^3_N &= &-2(N-1)\int_\RR \nu(dw) \beta_N(w) (1-\beta_N(w))- \frac{E - T}{E}.
\end{eqnarray*}
Then there a constant $C$ that depends only on $T, E$, and $M_4$ for which the following holds.
     \begin{equation*}
         \bar{K}^1_N+\bar{K}^2_N + \vert \bar{K}^3_N \vert \leq \frac{C(T,E,M_4)}{N}.
     \end{equation*}
\end{lem}

\begin{lem}[a version of Gr\"onwall's lemma]
\label{lem:Gronwall}
Let $u:\RR_+ \to \RR_+$ be an almost-everywhere differentiable function satisfying $\frac{du}{dt} \leq -au + bu^{1/2} + c$ for some constants $a>0$, $b\geq 0$ and $c\geq 0$. Then,
\[
u(t)
\leq 2u(0) e^{-at} +\frac{2c}{a} + \frac{4b^2}{a^2}
\quad \forall t \geq 0.
\]
\end{lem}

We are now ready to prove Theorem \ref{thm:UPoC}.

\begin{proof}[Proof of Theorem \ref{thm:UPoC}]
Call $h(t) = \EE[(V_t^1-Z_t^1)^2]$. Using Lemma \ref{lem:decouplingEW2} and exchangeability, we obtain
\begin{align*}
\EE [W_2(\bar{\mathbf{V}}_t, f_t)^2]
&\leq 2 \EE [W_2(\bar{\mathbf{V}}_t, \bar{\mathbf{Z}}_t)^2]
     + 2 \EE [W_2(\bar{\mathbf{Z}}_t, f_t)^2] \\
&\leq 2 \EE \left[ \frac{1}{N}\sum_{i=1}^N (V_t^i - Z_t^i)^2 \right]
     + \frac{C}{N^{1/3}} \\
&= 2 h(t) + \frac{C}{N^{1/3}}.
\end{align*}
Thus, it suffices to prove that $h(t) \leq 2 e^{-\frac{\mu T}{E} t} h(0) + C N^{-1/3}$, because $h(0) = W_2(f_0^N, f_0^{\otimes N})^2$ by exchangeability and \eqref{eq:W2f0N}.

We thus study the evolution of $h(t)$. We have
\begin{equation}
\label{eq:dht_UPoC}
\partial_t h(t)
= S_t^K + S_t^T + S_t^R + S_t^D,
\end{equation}
where those four terms correspond to Kac interactions, thermostat interactions, rescaling, and drift, respectively. That is: $S_t^K$ comes from the $\P_i$ terms of \eqref{eq:sde_vi} and \eqref{eq:sde_zi}; $S_t^T$ comes from the $\Q_i$ terms; $S_t^R$ comes from the summation in \eqref{eq:sde_vi} (and no term from \eqref{eq:sde_zi}); and $S_t^D$ comes from the drift term in \eqref{eq:sde_zi} (and no term from \eqref{eq:sde_vi}). To write them explicitly, let us first shorten notation: call $V_t^\ii = V_t^{\ii(\xi)}$, $Z_t^\ii = Z_t^{\ii(\xi)}$, and $F_t^1 = F_t^1(\mathbf{Z}_{t^\-},\xi)$. For clarity, we split the remaining of the proof into small steps.

\paragraph{Step 1, Kac term:}
Recall that the intensity of $\P_1(dt,d\theta,d\xi)$ is $\frac{2 \lambda dt d\theta d\xi \ind_{\{\ii(\xi)\neq 1\}}}{2\pi(N-1)}$, thus from \eqref{eq:sde_vi} and \eqref{eq:sde_zi}, using It\^{o} calculus, for $S_t^K$ we obtain:
\begin{align}
    \notag
    S_t^K
    &= \EE \int_0^{2\pi} \int_1^N
    \left[ \left(V_t^1 \cos\theta - V_t^{\ii} \sin\theta
    - Z_t^1 \cos\theta + F_t^1 \sin\theta \right)^2 - (V_t^1 - Z_t^1)^2
    \right] \frac{2\lambda d\theta d\xi}{2\pi(N-1)} \\
    \notag
    &= \EE \int_0^{2\pi} \int_1^N
    \left[ (V_t^1 - Z_t^1)^2 (\cos^2\theta -1)  +  (V_t^\ii - F_t^1)^2 \sin^2\theta
    \right] \frac{2\lambda d\theta d\xi}{2\pi(N-1)} \\
    \label{eq:StK_prelim}
    &= 2\lambda \left[
    -\frac{1}{2}h(t) + \frac{1}{2} \EE \int_1^N (V_t^\ii - F_t^1)^2 \frac{d\xi}{N-1}
    \right],
\end{align}
where in the second equality we discarded the cross-term because $\int_0^{2\pi} \cos\theta \sin\theta d\theta = 0$. Call $a(t) = \EE \int_1^N (Z_t^\ii - F_t^1)^2 \frac{d\xi}{N-1}$, thus $a(t) = \EE[W_2(\bar{\mathbf{Z}}_t^1, f_t)^2]$ thanks to \eqref{eq:EW2ftZi}; also note that $\EE \int_1^N (V_t^\ii - Z_t^\ii)^2 \frac{d\xi}{N-1} = \EE \frac{1}{N-1} \sum_{j=2}^N (V_t^j - Z_t^j)^2 = h(t)$ by exchangeability. For the integral in \eqref{eq:StK_prelim}, we add and subtract $Z_t^\ii$, thus obtaining:
\begin{align*}
    \EE \int_1^N (V_t^\ii - F_t^1)^2 \frac{d\xi}{N-1}
    &= h(t) + a(t) + 2\EE \int_1^N (V_t^\ii - Z_t^\ii)(Z_t^\ii - F_t^1) \frac{d\xi}{N-1} \\
    &\leq h(t) + a(t) + 2 h(t)^{1/2} a(t)^{1/2},
\end{align*}
where we have used the Cauchy-Schwarz inequality. Plugging this into \eqref{eq:StK_prelim} gives
\begin{equation}
    \label{eq:StK}
    S_t^K
    \leq \lambda a(t) + 2\lambda h(t)^{1/2} a(t)^{1/2}.
\end{equation}

\paragraph{Step 2, thermostat term:}
Recall that the intensity of $\Q_1(dt,dw)$ is $\mu dt \gamma(dw)$. Thus, again from \eqref{eq:sde_vi} and \eqref{eq:sde_zi}, we have for $S_t^T$:
\begin{align}
    \notag
    S_t^T
    &= \mu \int_\RR \left[(\beta_N(w)w - w)^2 - \EE (V_t^1-Z_t^1)^2 \right] \gamma(dw) \\
    \label{eq:StT}
    &= -\mu h(t)  + \mu \int_\RR (\beta_N(w)-1)^2 w^2 \gamma(dw).
\end{align}

\paragraph{Step 3, rescaling term:} Noting that every $\Q_j(dt,dw)$ has the same intensity $\mu dt \gamma(dw)$,
again from \eqref{eq:sde_vi} and \eqref{eq:sde_zi}, we obtain:
\begin{align}
    \notag
    S_t^R
    &= \mu \sum_{j=2}^N \EE \int_\RR \left[
    (\beta_N(w)V_t^1 - Z_t^1)^2 - (V_t^1 - Z_t^1)^2
    \right] \gamma(dw) \\
    \notag
    &= \mu (N-1) \EE \int_\RR \left[
    (\beta_N(w)V_t^1 - Z_t^1)^2 - (V_t^1 - Z_t^1)^2
    \right] \gamma(dw) \\
    \notag
    &= \mu(N-1)\EE \int_\RR \left[
    \beta_N(w)^2 (V_t^1-Z_t^1)^2 + (\beta_N(w)-1)^2 (Z_t^1)^2 \right. \\
    \notag
    & \qquad \qquad \qquad \qquad \left. {} + 2 \beta_N(w) (\beta_N(w)-1) (V_t^1-Z_t^1) Z_t^1
    - (V_t^1-Z_t^1)^2
    \right] \gamma(dw) \\
    \begin{split}
    \label{eq:StR}
    &= \mu(N-1) \int_\RR \left[
    (\beta_N(w)^2-1) h(t) + (\beta_N(w)-1)^2 E \right. \\
    & \qquad \qquad \qquad \qquad \left. {} + 2 \beta_N(w) (\beta_N(w)-1) \EE[(V_t^1-Z_t^1) Z_t^1]
    \right] \gamma(dw),
    \end{split}
\end{align}
where we have used that $\EE[(Z_t^1)^2] = E$ for all $t\geq 0$.

\paragraph{Step 4, drift term and conclusion:} It is clear  that $S_t^D = -2A \EE[(V_t^1-Z_t^1) Z_t^1]$. Plugging this, \eqref{eq:StK}, \eqref{eq:StT}, and \eqref{eq:StR} into \eqref{eq:dht_UPoC} and grouping terms, yields
\begin{equation}
\label{eq:dht_UPoC_conclusion}
\partial_t h(t)
\leq 2\lambda a(t)^{1/2} h(t)^{1/2}
    + K_N^1 h(t) + [\lambda a(t) + K_N^2]
    + K_N^3 \EE[(V_t^1-Z_t^1) Z_t^1],
\end{equation}
where $K_N^1$, $K_N^2$ and $K_N^3$ are given by
\begin{align*}
    K_N^1 &= - \mu + \mu (N-1) \int_\RR (\beta_N(w)^2 -1) \gamma(dw), \\
    K_N^2 &= \mu \int_\RR (\beta_N(w)-1)^2 w^2 \gamma(dw)
            + E \mu (N-1) \int_\RR (\beta_N(w)-1)^2 \gamma(dw), \\
    K_N^3 &= -2A + 2\mu (N-1) \int_\RR \beta_N(w) (\beta_N(w)-1) \gamma(dw).
\end{align*}
As in the proof of Lemma \ref{lem:contraction_ps}, a straightforward computation shows that $K_N^1 = -\mu G_N$, where $G_N = \int \frac{N w^2 \gamma(dw)}{NE-E+w^2}$. Since $T = \int w^2 \gamma(dw)$, we obtain
\[
    K_N^1
    = - \frac{\mu T}{E} + \mu \int_\RR \left[\frac{w^2}{E} - \frac{N w^2}{NE-E+w^2} \right] \gamma(dw)
    \leq - \frac{\mu T}{E} + \frac{C}{N}.
\]
Using the notation and results of Lemma \ref{lem:app}, we have $K_N^2 = \mu \bar{K}^2_N + E \mu \bar{K}^1_N$ and $K_N^3 = \mu \bar{K}^3_N$; thus, $K_N^2 \leq C/\sqrt{N}$ and $|K_N^3| \leq C/\sqrt{N}$. Also, $|\EE[(V_t^1-Z_t^1) Z_t^1]| \leq 2E$ and $h(t) \leq 4E$ because $\EE[(V_t^1)^2] = \EE[(Z_t^1)^2] = E$ for all $t\geq 0$. Finally, since $\int |v|^r f_0(dv) < \infty$ and condition \eqref{eq:unif-r-moment} holds for this $r$, we can apply Lemma \ref{lem:decouplingEW2} and obtain $a(t) = \EE[W_2(\bar{\mathbf{Z}}_t^1, f_t)^2] \leq C/N^{1/3}$ for all $t\geq 0$, in the case $r>4$. Using all these estimates in \eqref{eq:dht_UPoC_conclusion} gives
\begin{align*}
    \partial_t h(t)
    &\leq \frac{C}{N^{1/6}} h(t)^{1/2} - \frac{\mu T}{E} h(t)
    + \frac{C}{N} h(t) + \frac{C}{N^{1/3}} + \frac{C}{N} \\
    &\leq \frac{C}{N^{1/6}} h(t)^{1/2} - \frac{\mu T}{E} h(t) + \frac{C}{N^{1/3}}.
\end{align*}
Using Lemma \ref{lem:Gronwall} yields $h(t) \leq 2 h(0) e^{-\frac{\mu T}{E} t} + C N^{-1/3}$, as desired. The case $r\in(2,4]$ is treated similarly, using the other estimates of Lemma \ref{lem:decouplingEW2}.
\end{proof}

To conclude this section, we now prove Theorem \ref{thm:chaoticSN}, regarding the construction of a chaotic sequence supported on the sphere $S^{N-1}(\sqrt{NE})$.

\begin{proof}[Proof of Theorem \ref{thm:chaoticSN}]
We use the scaling procedure given by \eqref{eq:fN}, following \cite{cortez2016,fournier-guillin2015}. First, note that for any probability measures $\nu$ and $\tilde{\nu}$ on $\RR$, we have
\begin{equation}
\label{eq:W2nuq}
    W_2(\nu,\tilde{\nu})^2
    \geq (q^{1/2}-\tilde{q}^{1/2})^2,
\end{equation}
where $q = \int v^2 \nu(dv)$ and $\tilde{q} = \int v^2 \tilde{\nu}(dv)$. Indeed: for any coupling $X \sim \nu$ and $\tilde{X} \sim \tilde{\nu}$, we have $\EE[(X-\tilde{X})^2] = q + \tilde{q} - 2\EE(X\tilde{X}) \geq (q^{1/2}-\tilde{q}^{1/2})^2$, from which \eqref{eq:W2nuq} follows. Now, as specified in \eqref{eq:fN}, consider a random vector $\mathbf{X} = (X_1,\ldots,X_N)$ with distribution $f^{\otimes N}$ and define $Q = \frac{1}{N} \sum_i X_i^2$. Without loss of generality, we assume that $f$ does not have an atom at 0, thus $Q>0$ almost surely, and let $\mathbf{Y} = (Y_1,\ldots,Y_N)$ be given by $Y_i = (E/Q)^{1/2} X_i$ for all $i\in\{1,\ldots,N\}$ (if $f$ has an atom at 0, on the event $Q=0$, simply define $\mathbf{Y}$ as some fixed, independent, exchangeable random vector $\mathbf{Z}$ taking values on $S^{N-1}(\sqrt{NE})$; the following computations can be easily extended to this case). We define $f^N$ as the law of $\mathbf{Y}$. By construction, $f^N$ is supported on the sphere $S^{N-1}(\sqrt{NE})$. We have:
\[
    \frac{1}{N} \sum_{i=1}^N (Y_i - X_i)^2
    = \frac{1}{N} \sum_{i=1}^N X_i^2 \left( \frac{E^{1/2}}{Q^{1/2}} - 1\right)^2 
    = \left( \frac{E^{1/2}}{Q^{1/2}} - 1\right)^2 Q
    = (E^{1/2} - Q^{1/2})^2.
\]
From \eqref{eq:W2nuq} with $\nu = f$ and $\tilde{\nu} = \bar{\mathbf{X}} = \frac{1}{N}\sum_i \delta_{X_i}$, we obtain $\frac{1}{N} \sum_i (Y_i-X_i)^2 \leq W_2(f,\bar{\mathbf{X}})^2$. Since $\mathbf{Y} \sim f^N$ and $\mathbf{X}\sim f^{\otimes N}$, taking expectations gives $W_2(f^N,f^{\otimes N})^2 \leq \varepsilon_N(f)$. The explicit rates now follow directly from \eqref{eq:eps}.
\end{proof}

%%%%%%%%%%%%%%%%
% Conclusion
%%%%%%%%%%%%%%%%

\section*{Conclusion}

In this paper, we considered the thermostated Kac model in \cite{bonetto-loss-vaidyanathan2014} and introduced a global thermostat in the form of a rescaling mechanism, in order to restore the total energy. We used the coupling technique in \cite{cortez-fontbona2016} to obtain a quantitative, uniform in time propagation of chaos result for our particle system with rescaling factor $\beta_N$ given by \eqref{eq:betaN}. We studied the contractivity of the corresponding kinetic equation \eqref{eq:PDE} and its stationary distribution $f_\infty$. We discovered that the behavior of $f_\infty$ at the origin shows a phase transition from blowing up, to being continuous, $C^k$, and all the way to analytic based on the parameters $E,T$ and $\mu, \lambda$. We showed equilibration for our particle system in the $2$-Wasserstein distance in Lemma \ref{lem:equilibration_ps}. We note that the particular form of the Maxwellian thermostat $\gamma$ was not needed in our proofs; we only required it to have finite moments of sufficiently large order ($4$ is enough, see Lemma \ref{lem:app}). Thus Theorem \ref{thm:UPoC} can be generalized to other distributions for the thermostat; the same is true for parts of Theorem \ref{thm:continuity}.
Also, as a by-product of our study, in Theorem \ref{thm:chaoticSN} we showed how one can construct an $f$-chaotic sequence on the sphere $S^{N-1}(\sqrt{NE})$ by scaling an $f^{\otimes N}$-distributed random vector back to the sphere. This procedure gives a quantitative rate of chaoticity without requiring $f$ to have a density, and assuming only finite moment of order $2+\epsilon$. One could also study \emph{entropic chaoticity} for this sequence, as is done for the sequence obtained by conditioning $f^{\otimes N}$ to the sphere in \cite{carlen-carvalho-leroux-loss-villani2010,carrapatoso2015}; this is the subject of future research.

As mentioned earlier, we remark that, with some modifications,
such as conforming $\beta_N$ in \eqref{eq:betaN} to the weak thermostat as
\[
\sqrt{\frac{N E}{ (N-1) E + ( \mbox{sgn}(V_i) \sqrt{E}\cos\theta-w\sin\theta)^2 } }
\]
and replacing equation \eqref{eq:thermostatedBoltzmannKac} by
\[ \partial_t f_t(v) = 2\lambda \left( B_2[f_t, f_t] - f_t \right) + \mu \left( B_2[f_t ,\gamma] - f_t\right) - \frac{\mu}{2} \frac{E-T}{2E} \partial_v (v f_t(v) ),\]
our proof for contraction and propagation of chaos also work
for the more physical weak thermostat introduced in \cite{bonetto-loss-vaidyanathan2014} with the same rates. In this case, the proof for Lemma \ref{lem:app} becomes much more technical.

Many interesting questions remain open. For instance, one can consider the particle system with the exact rescaling factor $\alpha_N$, given by \eqref{eq:alphaN}. We do not know if propagation of chaos holds in this case; we hope that Theorem \ref{thm:UPoC} can be useful in answering this question. A related problem is the study of the equilibration properties of the $N$ particle system using more conventional metrics such as the spectral gap and relative entropy. It is possible that Kac's conjecture of a spectral gap bounded below uniformly in $N$ holds for this particle system, as well as the one with rescaling $\beta_N$. Not knowing the exact form of the equilibrium distribution in either case presents an additional difficulty for studying equilibration via spectral gap or relative entropy. We hope that global thermostats can be applied to study models  that have multiple blocks of Kac-type systems. Examples of such models are provided in \cite{bonetto-loss-tossounian-vaidyanathan2017} and \cite{tossounian-vaidyanathan2015}. Within each block, one can introduce a global thermostat to moderate the total energy of that block.

\bigskip

\paragraph{Acknowledgements.}
We would like to thank Federico Bonetto and Joaquin Fontbona for fruitful discussions during our work.

%%%%%%%%%%%%%%%%
% Appendix
%%%%%%%%%%%%%%%%

\section{Appendix}
\subsection{The master and kinetic equations}
\label{sec:ChaosLN}
In this section we justify the claim that \eqref{eq:PDE} is the only candidate as the propagation of chaos limit for our model with thermostat and rescaling (for both choices: \eqref{eq:alphaN} and \eqref{eq:betaN}).

It is expedient to use the generators of the Markov processes and the corresponding master equations which we now introduce.  The distribution $f_t^N(\mathbf{v})$ in Kac's original particle system evolved via the following \emph{master equation}:
\[
    \frac{\partial f_t^N}{\partial t} = N \lambda (Q - I) f_t^N
\]
where $Q$ is the Kac collision operator:
\[
    Q= \binom{N}{2}^{-1}\sum_{1\leq i< j \leq N} Q_{i,j}
\]
and where the operator $Q_{i,j}$ represents the collision of particles $i$ and $j$:
\[
    Q_{i,j}[f_t^N](\mathbf{v}) = \int_0^{2\pi} f^N_t(v_1, \dots,v_{i-1}, v_i'(\theta), \dots, v_{j-1}, v_j'(\theta), \dots ,v_N) \frac{d\theta}{2\pi}, 
\]
and where $I$ is the identity operator: $I[f^N] = f^N$.

The operator $P_i$ for the strong thermostat acting on particle $i$ is given by
\[
    P_i[f^N](\mathbf{v}) = \int_\RR f^N(v_1,\dots, v_{i-1}, w, \dots, v_N) dw \gamma(v_i). 
\]
Now, we introduce $\tilde{P}_j$ (resp. $\tilde{P}_j^\alpha$) which encorporates the action of strongly thermostating particle $j$ followed by the rescaling by the factor $\beta_N$ in \eqref{eq:betaN} (resp. $\alpha_N$ in \eqref{eq:alphaN}).
for any continuous and bounded function $\phi(\mathbf{v})$ we have
\[
\int_{\RR^N} \tilde{P}_1 [f^N_t](v) \phi(v) d\mathbf{v} = \int_{\RR^N} \int_\RR f^N_t(\mathbf{v}) \phi\left(\beta_N(w)(w,v_2, \dots, v_N)\right) \gamma(dw) d\mathbf{v},
\]
and
\[
\int_{\RR^N} \tilde{P}_1^\alpha [f^N_t](v) \phi(v) d\mathbf{v} = \int_{\RR^N} \int_\RR f^N_t(\mathbf{v}) \phi\left(\alpha_N(w, \mathbf{v})(w,v_2, \dots, v_N)\right) \gamma(dw) d\mathbf{v}.
\]
Thus, the generator for our particle system and its associated Master equation are the following.
\[
   L_N= N\lambda(Q-I) + \mu \sum_{j=1}^N (\tilde{P}_j -I ),
\]
\[
    \frac{\partial f^N_t}{\partial t} = L_N[f^N_t].
\]
Similarly, the generator $L_N^\alpha$ for the particle system with the scaling $\alpha_N$ and the associated Master equation are given by:
\[
   L_N^\alpha[f^N]= N\lambda(Q- I) + \mu \sum_{j=1}^N (\tilde{P}_j^\alpha-I),
\]
and
\begin{equation}\label{eq:masterAlpha}
    \frac{\partial f^N_t}{\partial t} = L_N^\alpha [f^N_t].
\end{equation}
We now state and prove the following lemma which can be seen as a propagation of chaos result at the level of generators.

\begin{lem}[generator level propagation of chaos]
\label{lem:chaosLN} 
Let $(f^N(\mathbf{v}))_N$ be a sequence of symmetric probability densities, with $f^N \in L^1(S^{N-1}(\sqrt{NE}),\sigma_N)$, that is chaotic to $f_0$. Assume that for some $r>2$ we have
\begin{equation}\label{eq:moment.condition}
    \sup_N \int_{S^{N-1}} f^N(\mathbf{v}) \vert v_1\vert^r \sigma_N(dv) <\infty.
\end{equation}
Then, for any $l \in \NN$ and any $h\in C^2(\RR^l)$ with 
\[
   \Vert h\Vert_\infty + \sum_{i=1}^l \Vert \partial_i h \Vert_\infty + \sum_{i,j=1}^l \Vert \partial_{i,j} h\Vert< \infty,
\]
the following holds:
\begin{equation}
\label{eq:generator}
\begin{split}
    \lim_{N \rightarrow \infty} \int_{S^{N-1}} L^\alpha_N [f^N] h\sigma_N(dv) &= \int_{\RR^{l}} 2\lambda \sum_{j=1}^l \left(\prod_{i \neq j} f_0(v_i)\right) [B_2[f_0,f_0]-f_0](v_j) h dv
    \\
    & \quad + \int_{\RR^{l+1}} f_0^{\otimes (l)} \left( \mu \sum_{j=1}^l (P_j^\ast-I) + \mu \frac{E -T} {2 E} v.\nabla\right)[h] dv.
\end{split}
\end{equation}
 The limit in \eqref{eq:generator} holds if we replace $(S^{N-1}(\sqrt{NE}), \sigma_N)$ by $(\RR^N, dv)$ and $L_N^\alpha$ by $L_N$.
\end{lem}

The idea of the proof is to divide the integral of the generator against a test function into a region where some components of $\mathbf{v}$ are controlled, and into the complementary region. The moment condition \eqref{eq:moment.condition} together with Tchebyshev's inequality will imply that the contribution from the complementary integral is small, while in the controlled region, Taylor's expansion will be effective.  
\begin{proof} We know from \cite{mckean1966} and \cite{bonetto-loss-vaidyanathan2014} that
\begin{align*}
  \lim_{N\rightarrow \infty} \int_{S^{N-1}}\hspace{-0.5cm}\sigma_N(dv) \left(N \lambda (Q-I)  + \sum_{j=1}^N (P_j -I) \right) [f^N](v) h(v) &=\\ \int_{\RR^{l}} 2\lambda \sum_{j=1}^l \left(\prod_{i \neq j} f_0(v_i)\right) [B_2[f_0,f_0]-f_0](v_j) h dv
  & +\mu \int_{\RR^k} \sum_{j=1}^l (P_j^\ast -I)h.
\end{align*}
Thus it suffices to study $\int_{S^{N-1}} \sum (\tp_j -P_j) [f^N]  h$. We have the following.
\begin{eqnarray*}
\int_{S^{N-1}} \sum (\tp_j -P_j) [f^N](v)  h(v)\sigma_N(dv) & = & \int_{S^{N-1}} f^N \sum_1^N (\tpast_j -P^\ast_j) [h] =: \mathcal{E}_1 + \mathcal{E}_2
\end{eqnarray*}
where $\mathcal{E}_1$ and $\mathcal{E}_2$ are given by
\[  
 \sum_{j=1}^l \int_{S^{N-1} \times \RR} f^N(v) \gamma(dw) [h(\alpha_N(v_{l+1}, w) (v_1, \dots, w, \dots,v_l))-h(v_1,\dots, w, \dots, v_l)]\sigma_N(dv),
\]
and
\[
(N-l) \int_{S^{N-1}\times\RR} f^N(v) [h(\alpha_N(v_{l+1},w) (v_1, \dots, v_l))-h(v)] \sigma_N(dv)\gamma(dw)
\]
respectively.
We will study $\lim_{N\rightarrow \infty}\mathcal{E}_2(N)$ only. Similar computations show that $\mathcal{E}_1(N)$ goes to zero.
Let $A_N$ be the region given below.
\begin{equation}\label{eq:ANappendix}
    A_N = S^{N-1}(\sqrt{NE})\times \RR \cap \left\{ (v,w) : \vert v_{l+1} \vert^2 \leq \frac{N E}{\log(N E)} \right\}.
\end{equation}
Then
\begin{eqnarray*}
    \mathcal{E}_2 & = & (N-l) \int_{A_N} f^N(v) [h(\alpha_N(v_{l+1},w) (v_1, \dots, v_l))-h(v)] \sigma_N(dv)dw + R_N.\\
\end{eqnarray*}
As promised, we now control the remainder term $R_N$ using Tchebyshev's inequality. Here we use the boundedness of the average $r^{\text{th}}$ moment.
\begin{eqnarray*}
\vert R_N \vert & \leq & 2 (N-l)\vert\vert h\vert\vert_\infty \times \left(\int_{A_N^c} f^N(v) \sigma_N(dv)\right)\\ & \leq & 2 (N-l) \vert\vert h\vert\vert_\infty \left(\frac{\log(NE)}{NE}\right)^{r/2} \int \vert v_{l+1}\vert^r f^N(v) \sigma_N(dv).
\end{eqnarray*}
This shows that we can neglect $R_N$. Let $\epsilon:= \alpha_N(v_{k+1},w)$. Using the following Taylor expansion:
\begin{equation*}
    h((1+\epsilon) v) - h(v) = \epsilon v. \nabla h(v) + \frac{\epsilon^2}{2} \sum_{i,j} v_i v_j \partial_{i,j}h(v_*)
\end{equation*}
we divide $\mathcal{E}_2-R_N$ into two parts $\mathcal{E}_{2A}$ ,coming from the  $\epsilon v. \nabla h(v)$ term, and $\mathcal{E}_{2B}$. We now prove that our assumptions on $(f^N)_N$ make $\mathcal{E}_{2B}(N)$ go to zero, and that $\mathcal{E}_{2A}$ gives the desired limit.
The expression:
\[
\mathcal{E}_{2B} \leq \frac{1}{2} \max_{i,j\leq l}\vert\vert\partial_{i,j} h \vert\vert_{\infty} \int_{A_N} \left(\sum_{i=1}^l v_i \right)^2 \frac{(N-l)f^N(v) \sigma_N(dv) \gamma(dw)}{ \left(\alpha_N(v_{l+1},w)^2 +1\right)^2} \left( \frac{v_{l+1}^2 -w^2}{NE +w^2 - v_{l+1}^2}\right)^2,
\]
together with \eqref{eq:ANappendix} and
the inequality $\sum_{i,j \leq l} \vert v_i v_j \vert \leq l \sum_{j=1}^l v_j^2$ show that 
\[
 \mathcal{E}_{2B} \leq C (l^2 E)  \max_{i,j\leq l}\vert\vert\partial_{i,j} h \vert\vert_{\infty} \left( \frac{1}{\log(N E)}+ \frac{ \int w^4 \gamma(dw)}{NE} \right).
\]
This implies that $\mathcal{E}_{2B}(N)$ goes to zero.
The term $\mathcal{E}_{2A}$ contributes to the expression \[
    \int_{S^{N-1}(\sqrt{NE})} L_N[f^N](\mathbf{v}) h(v_1,\dots, v_l) \sigma_N(dv)
\]
in the large $N$ limit. The coefficient of $v.\nabla h(v)$ in the limit of large $N$ behaves as follows:
\begin{equation*}
    \lim_{N\rightarrow \infty}\mu \frac{N-k}{ (\frac{N E}{NE+w^2- v_{l+1}^2}+1)} \frac{v_{l+1}^2 -w^2}{NE +w^2 - O(\frac{NE}{\log(NE )})} =\frac{\mu}{2}\frac{v_{l+1}^2 -w^2}{E}
\end{equation*}
Integrating against $w$ and against $v_{l+1}$ implies that, in the infinite $N$ limit, the coefficient of $v.\nabla h(v)$ is simply:
\[ \frac{\mu}{2} \frac{E -T}{E}.\]
Here we replaced $\int_{A_N}\int_\RR f^N \epsilon(v,w) v. \nabla h(v)$ by the integral over all of $S^{N-1}(\sqrt{NE}) \times \RR$. The cost of this change is negligible as seen above.
The same procedure shows that $\mathcal{E}_1$ becomes zero in the limit of large $N$. \end{proof}
\subsection{Proof of Lemma \ref{lem:Fourier}} \label{sec:Fourier}

\begin{proof}[Proof of Lemma \ref{lem:Fourier}]
For part (i), the fact that $y\in C^2(\RR)$ follows from $\int v^2 f_\infty(dv) = E$. Equation \eqref{eq:FourierInfinity} follows from \eqref{eq:PDEstationary} by choosing $\phi(v)= e^{-2\pi i v\xi}$, and noting that the left-hand side gives $2\pi i A \xi \int e^{-2\pi i v \xi} v f_\infty(dv) = -A \xi \frac{dy}{d\xi}(\xi)$.

We now prove (ii), that is, we need to show that $y(\xi)$ satisfies (P). To this end, let $(f_t)_{t\geq 0}$ be the weak solution to \eqref{eq:PDE} with $f_0 = \gamma$, and note that, since $f_t\to f_\infty$ weakly by Lemma \ref{lem:equilibrium}, we have $y(\xi) = \lim_{t\rightarrow\infty} \hat{f}_t(\xi)$ for all $\xi \in \RR$, where $\hat{f}_t(\xi) = \hat{g}_t(e^{A t }\xi)$ and $\hat{g}_t(\xi)$ satisfies the Fourier version of \eqref{eq:PDEweakNiceForm}, namely
\[
    \hat{g}_t = e^{-D t}\hat{\gamma} + e^{-D t}\int_0^t e^{D s}(2\lambda \hat{B_2}[\hat{g}_s,\hat{g}_s] +\mu \widehat{B_1[s]}) ds,
\]
with $D = 2\lambda + \mu$.

We now aim to prove that $\hat{g}_t$ satisfies (P) for every $t\geq 0$, for which we will use a fixed-point argument, similar to the one used in the proof of Theorem \ref{thm:WellPosedness2}. Consider the space
\[
F = \{z \in C(\RR,\mathbb{C}) : |z(\xi)| \leq 1, \forall \xi\in\RR \},
\]
endowed with the $L^\infty$ norm $\Vert \cdot \Vert_\infty$. For $\tau>0$ fixed, denote $\mathcal{X} = C([0,\tau],F)$, with the norm $|u| = \sup_{t\in[0,\tau]} \Vert u_t \Vert_\infty$, thus $(\mathcal{X},|\cdot|)$ is a complete metric space. For $u\in \mathcal{X}$, define
\[
\FF[u]_t
= e^{-D t}\hat{\gamma} + e^{-D t}\int_0^t e^{D s}(2\lambda \hat{B_2}[u_s,u_s] +\mu \widehat{B_1[s]}) ds.
\]
For $z\in F$, we have that $\hat{\gamma}$, $\widehat{B_1[s]}$ and $\hat{B}_2[z,z]$ all belong to $F$, which implies that $\FF$ maps $\mathcal{X}$ into itself. Also, since
\[
\Vert \hat{B}_2[z,z] -\hat{B}_2[w,w] \Vert_\infty \leq 2 \Vert z - w \Vert_\infty,
\]
it follows that $\FF$ is a contraction on $(\mathcal{X}, |\cdot|)$ when $\tau$ is small enough (the details are the same as in the proof of Theorem \ref{thm:WellPosedness2}). This proves that $\FF$ has a unique fixed point, which can then be extended to all times $t\geq 0$. By the uniqueness property, we deduce that $(\hat{g})_{t\geq 0}$ has to be this fixed point.

Now, consider the smaller space $F_P = \{ z\in F: \text{$z$ satisfies (P)}\}$ and $\mathcal{X}_P = C([0,\tau],F_P)$ with the same metrics as before, thus $\mathcal{X}_P$ is a closed subset of $\mathcal{X}$. We note that $\hat{B}_2[z,w]$ is always even (since replacing $\xi$ by $-\xi$ corresponds to changing $\theta$ into $\pi+\theta$.  This does not affect \eqref{eq:FourierB2}), and that $\hat{B}_2[z,z]$ is non-increasing on $[0,\infty)$ when $z$ satisfies (P): for $0\leq \xi_1 \leq \xi_2$,
\begin{eqnarray*}
\hat{B}_2[z,z](\xi_2)
& = & \int_0^{2\pi} z(\vert \xi_2\cos\theta \vert) z(\vert \xi_2\sin\theta\vert) \frac{d\theta}{2\pi} \\
& \leq & \int_0^{2\pi} z(\vert\xi_1\cos\theta\vert) z(\vert \xi_1\sin\theta\vert) \frac{d\theta}{2\pi}
= \hat{B}_2[z,z](\xi_1).
\end{eqnarray*}
Thus, $z \mapsto \hat{B}_2[z,z]$ preserves the property (P), which in turn implies that $\FF$ maps $\mathcal{X}_P$ into itself. Therefore, the fixed point $(\hat{g}_t)_{t\geq 0}$ of $\FF$ will belong to $\mathcal{X}_P$, because starting the Picard iterations from $u_t \equiv \hat{\gamma}$ will keep the sequence, and hence its limit, in $\mathcal{X}_P$. Consequently, $y(\xi)$ also satisfies (P) since $y(\xi)= \lim_{t\to \infty} \hat{g}_t(e^{At} \xi)$.

We now prove statement (iii).
The case $k=0$ holds since $y(\xi) \in C^2(\RR)$ and we have $f_\infty \in L^1(\RR)$, so that $\lim_{\xi \rightarrow \pm \infty} y(\xi) =0$ by the Riemann-Lebesgue lemma. For the cases $k\geq 1$, we use induction. Let $u_k(\xi) = \left(\xi \frac{d}{d\xi}\right)^k y$. Assume $u_k(\xi)$ is continuous and $\lim_{\xi\rightarrow \pm \infty} u_k(\xi) =0$ whenever $0\leq k \leq J$. We have the following Leibniz rule whenever $r$ and $m$ are less than $J$.
\begin{align*}
    \xi\frac{d}{d\xi} \int_0^{2\pi} u_r(\xi\cos\theta) u_m(\xi\sin\theta)\frac{d\theta}{2\pi} & = \int_0^{2\pi} u_{r+1}(\xi\cos\theta) u_m(\xi\sin\theta) \frac{d\theta}{2\pi}\\
    & \qquad + \int_0^{2\pi} u_r(\xi\cos\theta) u_{m+1} (\xi\sin\theta) \frac{d\theta}{2\pi}.
\end{align*}

This, together with equation \eqref{eq:FourierInfinity} gives:
\begin{eqnarray*}
     -A u_{J+1}(\xi) &= &2\lambda \sum_{k=0}^{J} \binom{J}{k} \int_0^{2\pi} u_{k}(\xi\cos\theta) u_{J-k}(\xi\sin\theta) \frac{d\theta}{2\pi}\\
       &  & - 2 \lambda u_J(\xi) + \mu \left( \xi \frac{d}{d\xi}\right)^J(\hat{\gamma}) - \mu u_J(\xi)
\end{eqnarray*}
The right-hand side is continuous, and converges to zero for large $\vert \xi\vert $ by the induction hypotheses and the fact that $\left( \xi \frac{d}{d\xi}\right)^J(\hat{\gamma})$ is a polynomial times a Gaussian.
\end{proof}

\subsection{Proof of Lemmas \ref{lem:GaussTrick} and \ref{lem:app}}
\label{sec:techincal}

\begin{proof}[Proof of Lemma \ref{lem:GaussTrick}]
First we prove the case when $r \in (0,1)$. Let 
\[\mathcal{G}[h](\xi) = r \hat{\gamma}(\xi) + (1-r) \int_0^{\frac{\pi}{2}} h(\xi \cos\theta)h(\xi\sin\theta)\frac{d\theta}{\pi/2},\]
and let $S$ be given by
\[S=\left\{ w \in C([0,\infty), [0,1]):  \sup_{\xi\neq 0} \frac{\vert w(\xi)-1 \vert}{\xi^{2}} < \infty\right\},\]
equipped with the norm $\Vert \cdot \Vert_S$ given by
\[\Vert f - g \Vert_S = \sup_{\xi\neq 0} \frac{\vert f(\xi)-g(\xi) \vert}{\xi^2}.\]
Then, by a standard $3\epsilon$ argument, we see that $(S, \Vert \cdot \Vert)_S$ is a complete metric space. It is also easy to show that $\mathcal{G}$ maps $S$ into $S$ and that $\mathcal{G}$ is a contraction in $(S, \Vert . \Vert_S)$ by a factor of $(1-r)$ (see  footnote \footnote{ This is because for all $\xi>0$ we have
\begin{eqnarray*}
\frac{\vert \mathcal{G}[f](\xi) - \mathcal{G}[g](\xi) \vert}{\xi^2} & \leq & (1-r) \int_0^{\frac{\pi}{2}} (f(\xi\cos\theta)+g(\xi\cos\theta))\frac{\vert (f(\xi\sin\theta)-g(\xi\sin\theta))\vert}{\xi^2}\\
 & \leq & (1-r) \Vert f - g \Vert_S \int_0^{\frac{\pi}{2}} 2 \xi^2 (\sin\theta)^2 \xi^{-2} = (1-r) \Vert f- g \Vert_S.
\end{eqnarray*}}).
It follows that $\hat{\gamma}$ is the unique fixed point of $\mathcal{G}$ in $S$.
We note also that $\mathcal{G}$ has a monotonicity property: if $u(\xi) \leq w(\xi)$ for all $\xi$, then $\mathcal{G}[u](\xi) \leq \mathcal{G}[w](\xi)$ for all $\xi$.
This implies that if $u(\xi)$ satisfies \eqref{eq:intineq}, then so does $\mathcal{G}[u]$. 
We are now ready to prove the lemma. Let $z(\xi)$ satisfy \eqref{eq:intineq}. Consider the sequence $\{z_n\}$ in $S$ with
\[ z_{k+1} = \mathcal{G}[z_k], \qquad z_0 = \max\{z, \hat{\gamma}\}.\]
We have $z_0 \in S$ since $z_0$ is continuous and $\sup_{\xi\neq 0} \vert z_0 - 1 \vert \xi^{-2}  \leq \sup_{\xi \neq 0} \vert \hat{\gamma}(\xi) - 1 \vert \xi^{-2} < \infty$.
We also note that $z_0$ satisfies \eqref{eq:intineq}. Since $z_0 \leq z_1$, we can show by induction and the monotonicity property of $\mathcal{G}$  that for each $\xi$ we have
\[ z(\xi) \leq z_1(\xi) \leq \dots \leq z_k(\xi) \leq \dots \leq 1.\]
By the contracting property of $\mathcal{G}$ we have that $z_\ast = \lim_{k\rightarrow \infty}z_k$ exists in $(S,\Vert \cdot\Vert_S)$. Thus, $\mathcal{G}(z_\ast) = z_\ast$ by continuity. Hence $z_\ast = \hat{\gamma}$ by the uniqueness of the fixed point. Therefore, we have $z(\xi) \leq z_0(\xi) \leq z_\ast(\xi) \leq \hat{\gamma}(\xi)$ for all $\xi$. This proves the lemma when $r\in (0,1)$.

When $r=0$, we argue by contradiction. Suppose $z(\xi)$ satisfies \eqref{eq:intineq} and \eqref{eq:zleqgamma}, but that $z(\xi_1)> \hat{\gamma}(\xi_1)$ for some $\xi_1$. Then $\xi_1 >0$. Since $z \leq 1$, $z(\xi_1) = \exp(-2 \pi^2 \bar{E} \xi_1^2)$ for some $\bar{E}< E$. Thus we have
\[ \frac{2}{\pi} \int_0^{\pi/2} z(\xi_1 \cos\theta) z(\xi_1 \sin\theta) \geq z(\xi_1) \geq  \frac{2}{\pi} \int_0^{\pi/2} \exp(-2 \pi^2 \bar{E} \xi_1^2 \cos^2\theta) \exp(-2 \pi^2 \bar{E} \xi_1^2\sin^2\theta).\]
So there is a subset $I_1$ of $[0,\pi/2]$ of positive measure for which
\[ \theta \in I_1 \Rightarrow \left( z(\xi_1 \cos\theta)\geq \exp(-2 \pi^2 \bar{E} \xi_1^2\cos^2\theta) \text{ or } z(\xi_1 \sin\theta)\geq \exp(-2 \pi^2 \bar{E} \xi_1^2\sin^2\theta) \right).\]
It follows that there is $\xi_2 \in (0, \xi_1)$ for which $z(\xi_2) \geq \exp(-2 \pi^2 \bar{E} \xi_2^2)$. Since $\xi_1$ was arbitrary, 
\[\inf\{ \xi>0: z(\xi) \geq \exp(-2 \pi^2 \bar{E} \xi_2)\} =0,\]
and there is a sequence $\eta_j>0$ approaching zero for which $z(\eta_j) \geq \exp(-2 \pi^2 \bar{E} \eta_j^2)$ for all $j$. Therefore
\begin{equation*}
    \liminf_{j\rightarrow \infty} (z(\eta_j)-1)\eta_j^{-2} \geq - 2\pi^2 \bar{E},
\end{equation*}
contradicting \eqref{eq:zleqgamma}.
\end{proof}

\begin{proof}[Proof of Lemma \ref{lem:app}] 
The proof relies on Taylor approximation.
Let
\[ I_N = \left\{ w: \left\vert \frac{E-w^2}{N E }\right\vert \leq \frac{1}{10}\right \}, \quad J_N = \RR - I_N.\]
We will need the observation that $\beta_N(w) \leq 2$ for all $N\geq 2$, and the fact that, when $\vert x\vert \leq \frac{1}{10}$, we have
$(1-x)^{-1/2} = 1 + \frac{x}{2} + \frac{3 }{8}x^2 \text{Er}(x)$; where $\vert \text{Er}(x) \vert \leq \left( \frac{10}{9}\right)^{5/2}$. 
We are now ready to find the desired upper bounds. We will let $x= \frac{E- w^2}{N E}$.
We note that when $w \in I_N$, we have $\vert x \vert \leq \frac{1}{10}$ and 
\begin{equation}\label{eq:Taylor}
   (N-1)(1-\beta_N(w)) = -\frac{N-1}{N} \left( \frac{E - w^2}{2 E} + \frac{3}{8} \text{Er}(\cdot) \frac{ (E - w^2)^2}{E^2 N} \right)
\end{equation}

Let us inroduce $\phi_N$ via the equation
\[
    \beta_N(w) = 1 + \frac{ \phi_N(w) }{N}.
\]
It follows that we have 
\begin{equation}
 \label{eq:phi_nu_integrable}
\int_\RR \vert \phi_N(w) \vert (1 + w^2) \nu(dv) \leq C (1 + M_4(\nu)),
\end{equation}
which follows from the fact that on $I_N$ we have
\begin{equation}
\label{eq:phi_bounded_by_energy}
\vert \phi_N \vert \leq  \left\vert \frac{E - w^2}{2 E} \right\vert \left(1 + \frac{3}{80} \left( \frac{10}{9}\right)^{5/2}\right),
\end{equation}
and the fact that, on $J_N$, we have $\vert\phi_N \vert \leq 3 N$. Thus
\[    
\int_{J_N} (1+w^2) \vert \phi_N \vert \leq 30 N \int_\RR (1 + w^2) \vert \left\vert \frac{E - w^2}{N E} \right\vert \nu(dv) \leq C(1 + M_4(\nu)). 
\]
We now tackle $\bar{K}_N^2$. Using the boundedness of $\beta_N$ and \eqref{eq:phi_nu_integrable}
\[
    \bar{K}_N^2  \leq 3 \int_\RR \nu(dw) \left\vert \frac{\phi(w)}{N} \right\vert w^2 \leq \frac{ C(1 + M_4(\nu))}{N}.
\]
We now tackle $\bar{K}_N^1$ and $\bar{K}_N^3$. Using Tchebyshev's inequality, we see that
\begin{align}
    \notag
    (N-1) \int_{J_N} (1- \beta_N)^2 \nu(dw) & \leq 9 (N-1) \int_{J_N} \nu(dw) \leq 9 (N-1) \int_{J_N} 100 \left\vert \frac{E- w^2}{E N} \right\vert^2 \nu(dw)\\  
    \label{eq:Tchebyshev}
    & \leq \frac{C}{N} (1 + M_4(\nu) ).
\end{align}
It follows that
\begin{align*}
    \notag
    \bar{K}_N^1 & = (N-1) \int_{I_N} \frac{\phi_N^2}{N^2} \nu(dw) + (N-1) \int_{J_N} (1- \beta_N)^2 \nu(dw) \\
        & \leq C \frac{1}{N} \int_\RR (1+w^2)^2 \nu(dw) + \frac{C}{N} (1 + M_4(\nu) ) \leq \frac{C}{N}(1+ M_4(\nu)).
\end{align*}
Here we used \eqref{eq:phi_bounded_by_energy} and \eqref{eq:Tchebyshev} for the first and second terms in the first inequality.

Finally, since $\bar{K}_3^N$ satisfies
\[
    \bar{K}_3^N = - 2(N-1) \int_\RR (1-\beta_N) \nu(dw) + 2 \bar{K}_N^1 - \frac{E-T}{E}
\]
the claim $\vert \bar{K}_3^N \vert \leq C(1 + M_4)/N$ follows from substituting \eqref{eq:Taylor} into the integral above and noting that each of 
\[ \int_{I_N} \frac{1}{N} \left\vert \text{Er}(\cdot) (E-w^2)^2 \right\vert \nu(dw), \quad
    \int_{J_N} \left\vert \frac{E- w^2}{2E}\right\vert \nu(dw), \quad \text{ and } 
    \int_{J_N} (N-1) \vert 1-\beta_N \vert \nu(dw)
\]
is bounded above by $C (1+ M_4(\nu) ) / N$.
\end{proof}

\subsection{Proof of the statement in Remark \ref{rmk:eps}}
\label{sec:app_Fournier}

Our case $r=4$ corresponds to the parameters $q=2p=4$ in \cite[Theorem 1]{fournier-guillin2013}. Since we are working in dimension $d=1$, we only need to consider the case $p> d/2$ (see step $2$ on page $717$ there). To extend the proof in \cite[Theorem 1]{fournier-guillin2013} to this case also, it remains to show the following inequality holds. 
\[    \sum_{n\geq 0} 2^{np} \min \left\{
    2^{-2p n} , (2^{-2pn}/N)^{1/2} \right\} \leq C N^{-1/2} \log N.
\]
This follows after we split the sum into the region $\{ n: \sqrt{N} < 2^{pn} \}$ and its complement,
 and noting that in this region $\min \{ 2^{-2p n} , (2^{-2pn}/N)^{1/2} \} = 2^{-2np}$, and the 
complementary region has O($\log N$) terms.

%%%%%%%%%%%%%%%
% References
\bibliographystyle{plain}
\bibliography{references.bib}{}
%%%%%%%%%%%%%%%

\end{document}